\documentclass[poms,final,nonblindrev]{poms1_V1}

\OneAndAHalfSpacedXI

\makeatletter

\RRHFirstLine{}
\RRHSecondLine{\bf\theRUNAUTHOR: \it\theRUNTITLE\vspace{1mm}}

\LRHFirstLine{}
\LRHSecondLine{\bf\theRUNAUTHOR: \it\theRUNTITLE\vspace{1mm}}

\renewcommand{\theARTICLETOPLEFT}{}
\renewcommand{\theARTICLETOPRIGHT}{}

\makeatother

\usepackage{graphicx}
\usepackage{subfigure, epsfig}
\usepackage{natbib}
\usepackage[inline]{enumitem}
\usepackage{comment}
\usepackage{booktabs}
\usepackage{subcaption}
\usepackage{multirow}

 \bibpunct[, ]{(}{)}{,}{a}{}{,}%
 \def\bibfont{\small}%
\TheoremsNumberedThrough     
\ECRepeatTheorems

\EquationsNumberedThrough    

\begin{document}


\RUNAUTHOR{Sahin, Kilic, and Yildirim}

\RUNTITLE{Accessibility-Dependent Maintenance Thresholds}

\TITLE{Condition-Based Maintenance of Degrading Assets under
Intermittent Accessibility}

\ARTICLEAUTHORS{%
\AUTHOR{Muhammet Ceyhan Sahin}
\AFF{Industrial \& Systems Engineering, Wayne State University, \EMAIL{ceyhan@wayne.edu}} 
\AUTHOR{Murat Yildirim}
\AFF{Industrial \& Systems Engineering, Wayne State University, \EMAIL{murat@wayne.edu}} 
\AUTHOR{Onur A. Kilic}
\AFF{Department of Operations, University of Groningen, \EMAIL{o.a.kilic@rug.nl}}
} 

\ABSTRACT{%
Many maintenance models implicitly assume that maintenance can be performed whenever intervention is warranted. In practice, however, environmental uncertainty, such as weather and sea conditions, can make maintenance opportunities intermittent and dynamically evolving. We consider a degrading asset for which preventive maintenance may be performed before failure, but intervention is possible only when the asset is accessible. Accessibility evolves stochastically over time and therefore affects not only whether maintenance can be performed now, but also the value of waiting for future opportunities. We formulate this setting as a finite-state Markov decision process under a long-run average cost criterion and investigate the structure of the optimal maintenance policy. Under monotone degradation and cost conditions, we show that an optimal policy retains a threshold form in asset condition, but unlike a single threshold, the optimal threshold varies with the accessibility state. Thus, intervention depends jointly on the asset condition and the expected evolution of future maintenance opportunities. Through a numerical study motivated by offshore wind turbine maintenance, we examine how accessibility dynamics, degradation characteristics, and economic parameters shape the optimal thresholds and compare the optimal policy with constant-threshold and age-based maintenance policies. In a representative setting, adapting the condition threshold to accessibility reduces long-run average cost by 4.31\% relative to an optimized condition-based maintenance policy that uses a single threshold under the same stochastic accessibility process. Relative to an optimized age-based policy evaluated under the same stochastic accessibility process, jointly using condition and accessibility information reduces long-run average cost by 33.95\%. The results show that maintenance thresholds are not universal: they should adapt to the degradation and accessibility characteristics of the system, and failing to account for these conditions can lead to substantial cost increases.
}%


\maketitle

%


\section{Introduction}\label{intro}

Capital-intensive assets degrade during operation, gradually losing performance and eventually reaching failure. Degradation is common to many physical assets, but its consequences are particularly important for capital-intensive systems because failures can result in substantial repair costs, production losses, and downtime. Preventive maintenance mitigates these consequences by intervening before failure occurs. 
Condition-based maintenance (CBM) improves the timing of preventive
intervention by using information about the asset's current condition,
rather than relying only on age or usage, to determine when maintenance
is economically justified \citep{jardine2006review,sun2023robust}.

A well-known structural feature of CBM models is the emergence of a condition-based threshold. Under appropriate degradation and cost conditions, the asset continues operating while its condition remains below a critical degradation level, and preventive maintenance becomes optimal once that level is reached \citep{elwany2011structured,si2018optimal}. This interpretation, however, implicitly relies on maintenance being possible when the threshold is reached. If intervention cannot be performed at that time, the decision problem changes fundamentally because the operator must account not only for the asset condition but also for when another maintenance opportunity may exist.

Such restrictions arise when maintenance opportunities depend on uncertain operating or environmental conditions. If an asset is currently accessible, the operator must decide whether to (i) use the present maintenance opportunity or (ii) continue operating and wait for a future opportunity whose timing is uncertain. Under the second option, the asset remains exposed to further degradation and may fail before another accessible period. Because corrective maintenance is subject to the same accessibility restrictions, such a failure can result in prolonged downtime until the asset becomes accessible. Thus, the value of performing preventive maintenance depends not only on the current degradation state, but also on the stochastic evolution of future maintenance opportunities.

Offshore wind provides a natural motivating example. Accessibility of an offshore wind turbine depends strongly on environmental conditions, particularly wind speed and wave height, and unfavorable conditions can restrict access for extended periods \citep{ren2021offshore}. A turbine may therefore reach a degradation level at which preventive maintenance is desirable while remaining inaccessible for intervention. Maintenance decisions must consequently account for both turbine condition and the evolution of the environmental conditions that determine future accessibility \citep{shafiee2019maintenance}. 

This interaction creates a dynamic tradeoff that is absent when turbine is always accessible. Consider a moderately degraded turbine that is currently accessible. If accessibility is likely to worsen and another maintenance opportunity may not become available for some time, it may be optimal to intervene immediately. In contrast, if favorable environmental conditions are likely to persist or return soon, continued operation may be preferred because another maintenance opportunity is expected to exist. The condition level at which preventive maintenance becomes desirable therefore need not be fixed: it can depend on the environmental state that governs both current accessibility and the likelihood of future maintenance opportunities.

Prior wind turbine maintenance studies have considered weather effects in several ways. \cite{byon2010optimal} model adverse weather through stochastic maintenance disruptions that delay intervention and increase downtime-related losses, while \cite{byon2010season} allow operating and maintenance conditions to vary across seasons. \cite{zhang2019opportunistic} consider weather-related maintenance waiting times and the associated production losses. Other studies incorporate broader effects of weather on turbine operation, including degradation, accessibility, and power production \citep{zheng2020optimal,zhu2019dynamic}, while related reliability models consider external environments that directly alter degradation rates \citep{ulukus2012optimal,zhang2013optimal}. These studies demonstrate several important ways in which environmental conditions affect maintenance systems. Our focus is different: we isolate accessibility as a stochastic process and examine how its evolution changes the condition level at which intervention becomes optimal.

Offshore wind also illustrates an important way in which stochastic
accessibility can arise. Environmental conditions evolve both seasonally and randomly. Some parts of the year provide substantially more favorable
accessibility than others, yet considerable uncertainty remains within each season. Consequently, the value of a given environmental state depends not only on whether maintenance can be performed now, but also on where the system lies within the seasonal accessibility cycle and how likely favorable conditions are to follow. The same turbine condition and the same accessibility status may therefore lead to different maintenance decisions when the outlook for future opportunities differs. This motivates the broader question studied in this paper: How should CBM decisions adapt when maintenance opportunities evolve stochastically over time?

We consider a degrading asset that evolves from an as-good-as-new condition through progressively degraded states and eventually to failure. At each decision epoch, the asset condition and the current environmental state are observed. Preventive maintenance may be performed before failure, while corrective maintenance is performed after failure; both interventions restore the asset to an as-good-as-new condition. Maintenance, however, can be performed only when the environmental state permits access. If the asset is inaccessible, intervention must be postponed while its degradation process continues to evolve. Because the environmental state also evolves stochastically, the current state contains information about the likelihood of future maintenance opportunities. The objective is to determine when to intervene so as to minimize long-run average cost.

We formulate the maintenance problem as a finite-state Markov decision process under a long-run average cost criterion. We characterize the structure of the optimal maintenance policy and analyze how maintenance decisions respond to changing accessibility. Our analysis shows that optimal maintenance decisions depend jointly on (i) the asset's current condition, (ii) the stochastic evolution of degradation, (iii) the current accessibility state, and (iv) the evolution of future maintenance opportunities.

The contributions of this paper are threefold:

\begin{itemize}
\item We extend the standard CBM framework by explicitly modeling accessibility as a separate stochastic process. To our knowledge, previous studies have not established the resulting optimal threshold structure when accessibility evolves as a distinct stochastic state. Asset degradation determines the physical need for intervention, while an exogenous environmental process determines whether maintenance can be performed. Modeling these processes jointly allows maintenance decisions to account for both degradation and the uncertainty of future maintenance opportunities.

\item We characterize the structure of optimal maintenance decisions under stochastic accessibility. We establish an accessibility-state-dependent threshold policy: for each state in which maintenance is feasible, there exists a condition threshold above which intervention is optimal. Hence, stochastic accessibility does not eliminate the familiar threshold structure of CBM, but makes the optimal threshold contingent on the current state of the accessibility process and its transition dynamics.

\item We quantify the value of adapting maintenance decisions to changing environmental conditions through a numerical study motivated by offshore wind. To identify this value, we benchmark the optimal policy against (i) an optimized condition-based policy that uses a single threshold and (ii) an optimized age-based policy both evaluated under the same stochastic accessibility process. We then examine how accessibility dynamics, degradation characteristics, and economic parameters affect the optimal thresholds. 
\end{itemize}

The remainder of the paper is organized as follows. Section~\ref{liter} reviews the related literature. Section~\ref{model} presents the maintenance model, and Section~\ref{struc_res} establishes the structural results under the long-run average cost criterion. Section~\ref{numeric} presents the numerical study, and Section~\ref{conclusion} concludes the paper. 

\section{Related Literature}\label{liter}

Maintenance optimization studies how inspection, maintenance, and replacement decisions should be made over time for degrading systems. Classical maintenance models focus on the timing of preventive or corrective interventions under stochastic degradation and failure, with objectives such as minimizing long-run average cost, discounted cost, downtime, or failure risk \citep{barlow1960optimum}. Early surveys by \cite{mccall1965maintenance}, \cite{pierskalla1976survey}, and \cite{sherif1981optimal} established the foundations of maintenance modeling for degrading systems. Subsequent research has examined the implementation of predictive maintenance, the integration of maintenance with production planning, and broader optimization-based approaches \citep{mckone2002guidelines,batun2012reassessing,
deJongeScarf2020review}. 

Within this broad literature, a substantial body of work considers CBM, in which information about the current condition of the asset is used to guide maintenance decisions \citep{jardine2006review,panagiotidou2010statistical,alaswad2017review}. Instead of replacing an asset solely according to its age or a fixed schedule, CBM relies on the state of the system, inferred from signals, measurements, or inspections to determine whether intervention is economically justified. 

Markov decision processes (MDPs) provide a natural framework for CBM optimization by linking maintenance decisions to the current condition of an asset, stochastic degradation, and future costs \citep{puterman2014markov}. Classical studies establish that, under suitable degradation and cost assumptions, optimal replacement decisions can be characterized by threshold policies \citep{kolesar1966minimum}. In particular, stochastic monotonicity of the degradation process, combined with appropriate cost conditions, can lead to a threshold structure: the asset continues operating while its condition is sufficiently good and is maintained or replaced once degradation reaches a critical level \citep{shahri2026data}. Such policies are intuitive and straightforward to implement, since maintenance decisions can be made by comparing the observed condition with a prescribed threshold.

Subsequent research extends this structural perspective to degrading systems with additional operational complexities. \cite{sloan2000combined} jointly optimize production and maintenance scheduling and examine structural properties of the optimal policy. \cite{elwany2011structured} and \cite{si2018optimal} establish control-limit replacement policies based on observed degradation signals, while \cite{liu2017condition} show that optimal maintenance thresholds can vary with system age when operating costs depend on both age and condition. \cite{drent2024condition} investigate production control in stochastically degrading systems with scheduled maintenance and establish monotonic properties of the optimal production policy. Similarly, \cite{wang2025learning} characterize optimal workload decisions under partially observed degradation and uncertainty about the relationship between workload and degradation. Together, these studies demonstrate how structural policy analysis can accommodate additional operational decisions, cost characteristics, and information limitations.

Another line of research examines how external conditions affect maintenance decisions. Existing studies can be grouped conceptually into two broad settings: those in which external conditions affect the degradation process, and those in which external conditions restrict when maintenance can be conducted or influence the consequences of delaying intervention.

In the first line of work, external conditions influence the asset degradation process. Reliability models have long recognized that operating and environmental conditions can affect degradation rates and failure behavior \citep{alaswad2017review, liang2023reliability}. These models assume that degradation results from asset's operation, but the rate of damage accumulation or failure risk depends on the operating environment. Thus, the environment affects the asset side of the problem: different environmental states can imply different degradation dynamics and, consequently, different maintenance or replacement thresholds. For example, \cite{ulukus2012optimal} address systems operating in randomly changing environments, where the current environmental condition determines how quickly degradation accumulates and leads to environment-specific replacement thresholds. \cite{zhang2013optimal} consider multi-component systems exposed to changing external conditions, where harsher environments accelerate component degradation and create stronger incentives for opportunistic maintenance. 

The second line of work considers external conditions that constrain maintenance opportunities. This is particularly relevant in offshore wind operations, where wind and sea conditions determine whether turbines can be accessed and maintenance activities can be performed \citep{ren2021offshore, martini2017accessibility}. \cite{byon2010optimal} study wind turbine maintenance under stochastic weather conditions, modeling adverse weather through exogenous probabilities that delays maintenance and increases downtime-related losses. In their setting, uncertainty in maintenance opportunities is represented through fixed disruption probabilities rather than through an explicitly evolving accessibility process. \cite{byon2010season} extend this setting by allowing operating and maintenance conditions to differ across seasons. They capture the seasonal variation in maintenance conditions, but do not represent accessibility as a Markov state whose evolution depends on the currently observed accessibility condition. \cite{byon2013wind} further considers wind-turbine maintenance decisions under uncertain weather conditions and develops tractable decision rules that account for weather-related maintenance delays and downtime costs. However, changing accessibility is not represented as an explicit stochastic process that directly governs the availability of maintenance opportunities over time. \cite{zhang2019opportunistic} consider stochastic weather restrictions by generating wind-speed trajectories and estimating the resulting maintenance waiting time, with the opportunistic maintenance threshold adjusted according to wind speed. Weather therefore affects maintenance duration and associated production losses rather than determining whether maintenance can be conducted in a given period. 

Collectively, these studies demonstrate how external conditions affect degradation, maintenance restrictions, delays, and production losses. Our study extends this literature by treating accessibility as a distinct stochastic process, separate from asset degradation, and examining its implications for optimal maintenance decisions. In particular, we establish how the stochastic evolution of accessibility modifies the classical CBM threshold structure under a long-run average cost criterion.

\section{Model}\label{model}
\subsection{Problem Setting and Sequence of Events}

We consider degrading asset that is reviewed at discrete decision epochs over an infinite horizon. A discrete-time formulation is appropriate because asset condition and accessibility information are typically assessed at regular planning intervals, such as days or weeks, and maintenance decisions are made based on the information available at these epochs. Offshore wind turbines provide the motivating application, where accessibility is governed by changing weather and sea conditions. 

At the beginning of each period, the asset condition and the prevailing accessibility condition are observed. The asset condition ranges from as-good-as-new through progressively degraded states to failure. The asset can be maintained preventively while it is functional or correctively after it has failed. If maintenance is not performed, a functional asset continues operating, and degrades over time, while a failed asset remains non-operational. The accessibility condition determines whether preventive or corrective maintenance can be commenced. 

Costs are incurred according to the state and the action taken in a period. Continuing operation without maintenance incurs an efficiency loss that depends on both the asset condition and the prevailing accessibility condition. This captures the economic consequences of degradation as well as differences in the production potential and hence generated revenue across accessibility conditions. Preventive and corrective maintenance incur fixed intervention costs that may also depend on the accessibility condition in which maintenance is commenced. These costs reflect both the direct maintenance expenditure and the revenue losses incurred while production is halted during maintenance operations, which may differ across accessibility conditions. Regardless of the accessibility condition, corrective maintenance is more costly than preventive maintenance.

At the end of each period, the system moves to the next decision epoch. If no maintenance is performed, 
asset condition evolves according to a stochastic degradation process. If maintenance is performed, the asset is restored to the as-good-as-new condition. Meanwhile, accessibility evolves independently as an exogenous stochastic process. In the offshore wind application, this accessibility process is induced by the evolution of weather and sea conditions and may exhibit both stochastic persistence and seasonal variation. The sequence of events in each period is as follows: 
\begin{enumerate*}[label=(\arabic*)]
    \item the asset and accessibility conditions are observed,
    \item the maintenance decision is made,
    \item costs are incurred, and
    \item the system evolves to the next state.
\end{enumerate*}
The process repeats indefinitely. The objective is to determine a maintenance policy that minimizes the long-run average cost.

\subsection{MDP Formulation}\label{mdp_form}

\subsubsection{State and Action Space}

Let $\mathcal{X}=\{0,1,\ldots,K\}$ denote the finite set of asset condition states, where \(x=0\) represents a new asset, state \(x=K\) represents a failed asset, and \(x=1,\ldots,K-1\) represent progressively degraded functional states. 
Similarly, let $\mathcal{W}=\{0,1,\ldots,W\}$ denote the finite set of accessibility states.
The system state is therefore defined as $s=(w,x)\in\mathcal{S}:=\mathcal{W}\times\mathcal{X}.$

Let $\mathcal{W}^{A}\subseteq\mathcal{W}$ denote the set of accessibility states in which the asset is accessible for maintenance. 
In each period with an accessibility state $w\in \mathcal{W}^{A}$, an action 
\[
a=
\begin{cases}
0 & \text{continue without maintenance}\\
1 & \text{perform maintenance}
\end{cases}
\]
is selected.
Hence, the state-dependent action set reads
\[
\mathcal{A}(w,x)=
\begin{cases}
\{0,1\}, & w\in\mathcal{W}^{A},\ x\in\{1,\ldots,K\},\\
\{0\}, & w\notin\mathcal{W}^{A},\ x\in\{0,\ldots,K\}.
\end{cases}
\]

\subsubsection{Transition Dynamics}

The system evolves as a discrete-time Markov decision process.
Let $\{(w_t,x_t)\}_{t\geq 0}$ denote the controlled state process, where \(t\) indexes the decision periods.
The accessibility state follows the transition matrix $P^W=(p^W_{wu})_{w,u\in\mathcal{W}},$ where $p^W_{wu}=\mathbb{P}(w_{t+1}=u\mid w_t=w).$ 
If maintenance is not commenced, asset condition follows the transition matrix $P^X=(p^X_{xy})_{x,y\in\mathcal{X}},$ where $p^X_{xy}=\mathbb{P}(x_{t+1}=y\mid x_t=x,a_t=0).$ 
If maintenance is commenced, asset condition defaults to state $x_{t+1}=0$.  
The asset degradation process and accessibility process are independent. 
Hence, the controlled transition probabilities are as follows:
\[
P((u,y)\mid(w,x),a)
=
\begin{cases}
p^W_{wu}p^X_{xy}, & a=0,\\[1mm]
p^W_{wu}\mathbf{1}_{\{y=0\}}, & a=1.
\end{cases}
\]

\subsubsection{Cost Structure}

Let \(g(w,x)\) denote the single-period efficiency loss incurred when no maintenance is performed in state \((w,x)\), where $0\leq g(w,x)<\infty$ for 
$(w,x)\in\mathcal{S}$.
For functional states $x\in\{0,1,\ldots,K-1\}$, it captures the economic consequences of continued operation, such as condition-dependent efficiency loss, reduced production, or other degradation-related penalties. For the failed state \(x=K\), it represents the economic consequence of asset failure and downtime. Its dependence on \(w\) allows these consequences to vary with prevailing accessibility conditions.

For \(w\in\mathcal{W}^{A}\), let \(c^{\text{PM}}(w)\) and \(c^{\text{CM}}(w)\) denote the total costs of preventive and corrective maintenance, respectively. These include the direct intervention costs as well as maintenance-related production losses, which depend on the accessibility. They satisfy
\[
0<c^{\text{PM}}(w) \leq c^{\text{CM}}(w)<\infty.
\]

The single-period cost is therefore
\[
c((w,x),a)
=
\begin{cases}
g(w,x), & a=0,\\
c^{\text{PM}}(w), & a=1,\ x\in\{1,\ldots,K-1\},\\
c^{\text{CM}}(w), & a=1,\ x=K.
\end{cases}
\]

\subsection{Long-Run Average Cost and Recurrence}
\label{subsec:lrac_recurrence}

A stationary deterministic policy is a mapping $\pi:\mathcal{S}\rightarrow\{0,1\}$ such that
\[
\pi(w,x)\in\mathcal{A}(w,x),
\qquad
(w,x)\in\mathcal{S}.
\]
For a policy \(\pi\), define the long-run average cost starting from state \((w,x)\) as
\[
J^\pi(w,x)
=
\limsup_{T\rightarrow\infty}
\frac{1}{T}
\mathbb{E}^{\pi}_{(w,x)}
\left[
\sum_{t=0}^{T-1}
c\big((w_t,x_t),\pi(w_t,x_t)\big)
\right].
\]

Because the problem is defined over an infinite horizon, we require the controlled process to have a common long-run recurrent behavior rather than separate recurrent regimes determined by the initial state. The following conditions ensure recurrent accessibility  to maintenance and a positive-probability path toward failure when maintenance is continually postponed.

\begin{assumption}[Recurrence conditions]
\label{ass:unichain}
The following conditions hold:
\begin{enumerate}[label=\Alph*.]

    \item
    The accessibility transition matrix \(P^W\) is irreducible on
    \(\mathcal{W}\).

    \item
    The set of accessible states is nonempty:
    $\mathcal{W}^{A}\neq\emptyset.$

    \item
    The failed state is absorbing under continued operation,
    $p^X_{KK}=1,$ 
    
    and for every \(x<K\), there exists \(y>x\) such that $p^X_{xy}>0.$

    \item
    A new asset can remain in the new condition for one period with positive probability: $p^X_{00}>0.$

\end{enumerate}
\end{assumption}

Assumption~\ref{ass:unichain}A prevents the accessibility process from decomposing into disconnected closed classes. Together with Assumption~\ref{ass:unichain}B, it ensures that accessible conditions can be reached from any current accessibility state. Assumption~\ref{ass:unichain}C provides a positive-probability degradation path from every non-failed state to failure if maintenance is continually postponed. Assumption~\ref{ass:unichain}D allows the asset to remain in the as-good-as-new state while accessibility evolves, which provides a common post-maintenance state that can be reached under every stationary policy. Together with the maintenance-reset mechanism, these properties imply the required unichain structure.

\begin{proposition}
\label{prop:unichain}
Under Assumption~\ref{ass:unichain}, every stationary deterministic policy considered in this paper induces a unichain Markov chain on \(\mathcal{S}\).
\end{proposition}

\begin{proof}{Proof.}
The accessibility process is irreducible, and by Assumptions~\ref{ass:unichain}A--B, an accessible state is reachable from every accessibility state. Under continued operation, Assumption~\ref{ass:unichain}C implies that the asset can eventually reach the failed state. Once the asset is failed and accessibility permits maintenance, corrective maintenance is mandatory and resets the asset condition to \(x=0\). Hence, all recurrent behavior is connected through the post-maintenance state structure, so the induced Markov chain has a single recurrent class, possibly together with transient states.
\end{proof}

The proposition ensures that the long-run behavior generated by a stationary policy is not determined by the initial state, see \cite[Prop.~8.2.1]{puterman2014markov}. Consequently, the long-run average cost under a fixed stationary deterministic policy is constant over \(\mathcal{S}\), and we write
\[
\eta_\pi
:=
J^\pi(w,x),
\qquad
(w,x)\in\mathcal{S}.
\]
The optimization problem is therefore
\[
\eta^*
=
\inf_{\pi}\eta_\pi,
\]
and an optimal policy satisfies
\[
\pi^*
\in
\arg\min_{\pi}\eta_\pi.
\]

For finite-state, finite-action unichain MDPs, the average-cost optimality equation admits a scalar optimal average cost and a relative value function, and a stationary deterministic selector attaining its minimum is average-cost optimal; see \cite[Thm.~8.4.5]{puterman2014markov}.

\section{Structural Results under long-run average cost}
\label{struc_res}

The finite-state MDP can be solved directly to determine an optimal action for every combination of asset condition and accessibility. Beyond computing these state-by-state decisions, we are interested in whether they exhibit a systematic structure. In particular, we ask whether worsening asset condition eventually makes maintenance preferable to continued operation and whether this decision can be represented by a condition threshold for each accessible state.

To establish such a structure, the ordering of the asset condition states must also be reflected in their future evolution and economic consequences. The conditions required for this analysis are introduced next.

\subsection{Conditions for Structural Analysis}
\label{cond}

First, degradation should preserve the physical ordering of asset condition states: an asset that is currently in a worse condition should not have systematically better future condition than one that is currently in a better condition.

\begin{assumption}[Monotone degradation]
\label{ass:monotone_asset_transition}
Under continued operation,
\[
p^X_{xy}=0,
\qquad
y<x,
\]
and \(P^X\) is stochastically monotone. That is, for every
\(x_2\geq x_1\) and \(k\in\mathcal{X}\),
\[
\sum_{y=k}^{K}p^X_{x_2y}
\geq
\sum_{y=k}^{K}p^X_{x_1y}.
\]
\end{assumption}

The first part of Assumption~\ref{ass:monotone_asset_transition} rules out spontaneous improvement in asset condition when no maintenance is performed. The stochastic-monotonicity condition is weaker than requiring degradation in every period: the asset may remain in its current condition, but a worse current state leads to a stochastically worse distribution of future condition states. Second, the economic consequence of continued operation should preserve the same ordering.

\begin{assumption}[Monotone efficiency loss]
\label{ass:monotone_operating_cost}
For each \(w\in\mathcal{W}\),
\[
g(w,x+1)\geq g(w,x),
\qquad
x=0,\ldots,K-1.
\]
\end{assumption}

Assumption~\ref{ass:monotone_operating_cost} states that degradation does not make continued operation economically more attractive. A higher condition state may represent lower efficiency, greater production loss, increased exposure to failure, or, at \(x=K\), complete asset unavailability. Finally, corrective maintenance should not be less expensive than preventive maintenance.

\begin{assumption}[Maintenance-cost ordering]
\label{ass:maintenance_cost_order}
For every \(w\in\mathcal{W}^{A}\),
\[
c^{\text{PM}}(w)\leq c^{\text{CM}}(w).
\]
\end{assumption}

This ordering reflects the additional consequences associated with failure and is used to preserve the ordering between the last degraded state and the failed state in the structural analysis.

\subsection{Average-Cost Optimality Equation}
\label{acoe}

Proposition~\ref{prop:unichain}, together with the finiteness of the state and action spaces, allows the maintenance problem to be characterized through the average-cost optimality equation. There exist a scalar \(\rho^*\) and a bias function $h:\mathcal{S}\rightarrow\mathbb{R}$ such that
\[
\rho^*+h(w,x)
=
\min_{a\in\mathcal{A}(w,x)}
\left\{
c((w,x),a)
+
\sum_{(u,y)\in\mathcal{S}}
P((u,y)\mid(w,x),a)h(u,y)
\right\}.
\]
Moreover, $\rho^*=\eta^*.$ To study the structure of the minimizing action, define the continuation value associated with waiting as
\[
Q^0(w,x)
=
g(w,x)
+
\sum_{u\in\mathcal{W}}
\sum_{y\in\mathcal{X}}
p^W_{wu}p^X_{xy}h(u,y).
\]
For preventive maintenance, define
\[
Q^P(w)
=
c^{\text{PM}}(w)
+
\sum_{u\in\mathcal{W}}
p^W_{wu}h(u,0),
\]
and for corrective maintenance,
\[
Q^C(w)
=
c^{\text{CM}}(w)
+
\sum_{u\in\mathcal{W}}
p^W_{wu}h(u,0).
\]

The optimality equation for \(x=0\) can therefore be written as
\[
\rho^*+h(w,0)
=
Q^0(w,0),
\qquad
w\in\mathcal{W},
\]
and, for \(x=1,\ldots,K-1\),
\[
\rho^*+h(w,x)
=
\begin{cases}
\min\{Q^0(w,x),Q^P(w)\},
    & w\in\mathcal{W}^{A},\\[1mm]
Q^0(w,x),
    & w\notin\mathcal{W}^{A}.
\end{cases}
\]
For the failed state,
\[
\rho^*+h(w,K)
=
\begin{cases}
Q^C(w),
    & w\in\mathcal{W}^{A},\\[1mm]
Q^0(w,K),
    & w\notin\mathcal{W}^{A}.
\end{cases}
\]

Thus, when the asset is degraded but functional and accessible, preventive maintenance is optimal whenever
\[
Q^P(w)\leq Q^0(w,x).
\]
The preventive-maintenance value \(Q^P(w)\) does not depend on the current degradation state because maintenance always restores the asset to the same new state. Consequently, if \(Q^0(w,x)\) can be shown to increase with degradation, the maintenance region must inherit an ordered structure. The next result establishes the required monotonicity.

\subsection{Monotonicity of the Bias Function}
\label{monoton}

The key structural step is to establish that a worse asset condition cannot have a lower relative long-run cost. Intuitively, degradation increases the current efficiency loss and leads to stochastically worse future conditions, whereas maintenance restores all degraded states to the same new condition. The following lemma formalizes this ordering.

\begin{lemma}
\label{lem:bias_monotone_x}
Under Assumptions~\ref{ass:unichain},
\ref{ass:monotone_asset_transition},
\ref{ass:monotone_operating_cost}, and
\ref{ass:maintenance_cost_order},
there exists a solution \((\rho^*,v)\) to the average-cost optimality equation such that
\[
v(w,x+1)\geq v(w,x),
\qquad
w\in\mathcal{W},
\quad
x=1,\ldots,K-1.
\]
\end{lemma}

\begin{proof}{Proof.}
For \(\beta\in(0,1)\), define
\[
(T_\beta v)(w,x)
=
\min_{a\in\mathcal{A}(w,x)}
\left\{
c((w,x),a)
+
\beta
\sum_{(u,y)\in\mathcal{S}}
P((u,y)\mid(w,x),a)v(u,y)
\right\}.
\]
Let
\[
v_0\equiv0,
\qquad
v_{n+1}=T_\beta v_n.
\]

Suppose
\[
v_n(w,x+1)\geq v_n(w,x),
\qquad
x=1,\ldots,K-1.
\]
For \(1\leq x_1\leq x_2\leq K-1\), Assumption~\ref{ass:monotone_asset_transition} gives
\[
\sum_{y\in\mathcal{X}}
p^X_{x_2y}v_n(u,y)
\geq
\sum_{y\in\mathcal{X}}
p^X_{x_1y}v_n(u,y),
\qquad
u\in\mathcal{W}.
\]
Hence,
\[
\sum_{u\in\mathcal{W}}
\sum_{y\in\mathcal{X}}
p^W_{wu}p^X_{x_2y}v_n(u,y)
\geq
\sum_{u\in\mathcal{W}}
\sum_{y\in\mathcal{X}}
p^W_{wu}p^X_{x_1y}v_n(u,y).
\]
By Assumption~\ref{ass:monotone_operating_cost},
\[
g(w,x_2)
+
\beta
\sum_{u,y}
p^W_{wu}p^X_{x_2y}v_n(u,y)
\geq
g(w,x_1)
+
\beta
\sum_{u,y}
p^W_{wu}p^X_{x_1y}v_n(u,y).
\]

For \(w\in\mathcal{W}^{A}\), the preventive-maintenance term
\[
c^{\text{PM}}(w)
+
\beta
\sum_{u\in\mathcal{W}}
p^W_{wu}v_n(u,0)
\]
is independent of \(x\). Therefore,
\[
(T_\beta v_n)(w,x+1)
\geq
(T_\beta v_n)(w,x),
\qquad
x=1,\ldots,K-2.
\]
The same inequality follows directly from the waiting term when
\(w\notin\mathcal{W}^{A}\).

For \(w\in\mathcal{W}^{A}\),
\begin{align*}
(T_\beta v_n)(w,K)
&=
c^{\text{CM}}(w)
+
\beta
\sum_{u\in\mathcal{W}}
p^W_{wu}v_n(u,0)
\\
&\geq
c^{\text{PM}}(w)
+
\beta
\sum_{u\in\mathcal{W}}
p^W_{wu}v_n(u,0)
\\
&\geq
(T_\beta v_n)(w,K-1),
\end{align*}
where the first inequality follows from
Assumption~\ref{ass:maintenance_cost_order}. For
\(w\notin\mathcal{W}^{A}\), stochastic monotonicity gives
\[
(T_\beta v_n)(w,K)
\geq
(T_\beta v_n)(w,K-1).
\]

Thus \(T_\beta\) preserves monotonicity on
\(\{1,\ldots,K\}\). Since \(v_0\equiv0\), induction yields
\[
v_n(w,x+1)\geq v_n(w,x),
\qquad
x=1,\ldots,K-1.
\]
Discounted value iteration converges to the unique discounted value function \(V_\beta\), so
\[
V_\beta(w,x+1)\geq V_\beta(w,x).
\]

Let \(\beta_m\uparrow1\). By the finite-state unichain average-cost limit result, there exist a subsequence, a reference state \(s_0\), a scalar \(\rho^*\), and a bias function \(v\) such that
\[
(1-\beta_m)V_{\beta_m}(s_0)\rightarrow\rho^*
\]
and
\[
V_{\beta_m}(w,x)-V_{\beta_m}(s_0)
\rightarrow v(w,x).
\]
Taking limits preserves the inequalities, yielding
\[
v(w,x+1)\geq v(w,x),
\qquad
x=1,\ldots,K-1.
\]
\end{proof}

Lemma~\ref{lem:bias_monotone_x} establishes the ordering needed for the maintenance decision. A more degraded asset has a weakly larger relative long-run cost because continued operation combines a larger current efficiency loss with stochastically worse future degradation. Preventive maintenance behaves differently: regardless of the current degraded condition, it restores the asset to the same new state. This distinction is what produces the threshold structure established next.

\subsection{Optimality of Accessibility-Dependent Threshold Policies}
\label{optimal}

We can now determine the structure of the optimal maintenance decision. For a fixed accessible state, the value of preventive maintenance is independent of the current degraded condition, whereas the value of continued operation becomes less attractive as degradation progresses. Hence, once preventive maintenance becomes preferable to continued operation, it cannot become less preferable at a worse condition state.

\begin{theorem}
\label{thm:threshold_optimal_lrac}
Under Assumptions~\ref{ass:unichain},
\ref{ass:monotone_asset_transition},
\ref{ass:monotone_operating_cost}, and
\ref{ass:maintenance_cost_order},
there exists an long-run average cost-optimal stationary deterministic policy
\(\pi^*\) such that, for every
\(w\in\mathcal{W}^{A}\), there exists
\[
\theta(w)\in\{1,\ldots,K\}
\]
satisfying
\[
\pi^*(w,x)
=
\begin{cases}
0, & x<\theta(w),\\
1, & x\geq\theta(w),
\end{cases}
\qquad
x=1,\ldots,K.
\]
For \(w\notin\mathcal{W}^{A}\),
\[
\pi^*(w,x)=0,
\qquad
x\in\mathcal{X}.
\]
\end{theorem}

\begin{proof}{Proof.}
Fix \(w\in\mathcal{W}^{A}\) and define
\[
\mathcal{M}_w
=
\left\{
x\in\{1,\ldots,K-1\}:
Q^P(w)\leq Q^0(w,x)
\right\}.
\]

For \(1\leq x_1\leq x_2\leq K-1\), Lemma~\ref{lem:bias_monotone_x} and Assumptions~\ref{ass:monotone_asset_transition}--\ref{ass:monotone_operating_cost} imply
\[
Q^0(w,x_2)\geq Q^0(w,x_1).
\]
Since \(Q^P(w)\) is independent of \(x\), \(\mathcal{M}_w\) is an upper set.

If \(\mathcal{M}_w\neq\emptyset\), define
\[
\theta(w)=\min\mathcal{M}_w.
\]
If \(\mathcal{M}_w=\emptyset\), define
\[
\theta(w)=K.
\]
Since maintenance is mandatory at \(x=K\) for
\(w\in\mathcal{W}^{A}\), an optimal selector of the average-cost optimality equation satisfies
\[
\pi^*(w,x)
=
\begin{cases}
0, & x<\theta(w),\\
1, & x\geq\theta(w).
\end{cases}
\]

For \(w\notin\mathcal{W}^{A}\),
\[
\mathcal{A}(w,x)=\{0\},
\qquad
x\in\mathcal{X},
\]
and therefore
\[
\pi^*(w,x)=0.
\]
\end{proof}

Theorem~\ref{thm:threshold_optimal_lrac} shows that changing accessibility preserves the familiar threshold structure of CBM, but makes the maintenance decision dependent on accessibility. For each accessibility state in which maintenance can be conducted, the optimal decision is summarized by a single condition threshold. Below the threshold, continued operation is preferred; once the asset reaches the threshold or a worse condition, maintenance is preferred. If the threshold equals \(K\), preventive maintenance is never performed and intervention occurs only after failure.

The thresholds need not be identical across accessible states. Current accessibility affects both the immediate consequences of the maintenance decision and the distribution of future accessibility conditions. Thus, the threshold reflects the tradeoff between using the current maintenance opportunity and postponing intervention while facing uncertain future accessibility. In accessibility states where the asset is not accessible, no threshold decision is required because maintenance cannot be conducted regardless of asset condition.

\subsection{Threshold Policy Representation}
\label{thresh}

Theorem~\ref{thm:threshold_optimal_lrac} implies that an optimal stationary policy can be represented using one condition threshold for each accessible state rather than a separate action for every state in
\(\mathcal{S}\).

Let
\[
\tau=(\tau_w)_{w\in\mathcal{W}^{A}},
\]
where
\[
\tau_w\in\{1,\ldots,K\},
\qquad
w\in\mathcal{W}^{A}.
\]
Define the feasible threshold set
\[
\Theta
=
\prod_{w\in\mathcal{W}^{A}}
\{1,\ldots,K\}.
\]

For a given \(\tau\in\Theta\), define
\[
\pi_\tau(w,x)
=
\begin{cases}
1,
    & w\in\mathcal{W}^{A},\ x\geq\tau_w,\\
0,
    & \text{otherwise}.
\end{cases}
\]

The threshold values have a direct interpretation. If
\[
\tau_w=1,
\]
maintenance is performed whenever the asset is degraded or failed and accessibility permits. If
\[
\tau_w=K,
\]
preventive maintenance is never performed in accessibility state \(w\), and maintenance occurs only after failure. No threshold is assigned to
\(w\notin\mathcal{W}^{A}\), since the asset is not accessible for maintenance in those states.

Let \(\eta_\tau\) denote the long-run average cost induced by policy \(\pi_\tau\). The optimal threshold vector satisfies
\[
\tau^*
\in
\arg\min_{\tau\in\Theta}
\eta_\tau.
\]
Thus, the structural result reduces the representation of the optimal stationary maintenance policy from arbitrary state-action decisions to a collection of condition thresholds associated with the accessible states.

\section{Numerical Study}\label{numeric}

We conduct numerical experiments to investigate how stochastic and seasonal accessibility affects optimal CBM decisions and to quantify the economic value of adapting maintenance thresholds to changing accessibility. We use offshore wind as the motivating application, where accessibility depends on weather and sea conditions that exhibit both annual seasonality and period-to-period uncertainty. 

We first examine a base-case setting to illustrate how optimal maintenance thresholds change throughout the year and how their variation reflects the expected evolution of future accessibility. We then evaluate the economic benefits of accessibility-dependent thresholds by comparing the optimal policy with optimized constant-threshold and age-based policies, both subject to the same stochastic accessibility process. A common-path simulation further illustrates how the different policies translate condition and accessibility information into maintenance decisions. Finally, sensitivity analyses examine how accessibility, degradation, and economic characteristics affect the level and seasonal variation of the optimal thresholds.

\subsection{Experimental Setting}

We next specify the degradation, accessibility, and cost parameters used in the numerical experiments. These choices are specific to the computational study. 

\subsubsection{Degradation Process}

The degradation-process parameters are selected to generate a monotone and stochastic degradation pattern consistent with gradual degradation in offshore wind turbines. Stochastic processes such as Wiener, Gamma, and inverse Gaussian processes are commonly used to represent physical degradation in maintenance models \citep{sun2023robust}. We use a Gamma process because its nonnegative increments naturally represent monotone degradation. Specifically, if $Z(t)$ denotes cumulative degradation, then
\[
Z(t+1)-Z(t)\sim \operatorname{Gamma}(\kappa,\nu),
\]
where $\kappa>0$ is the shape parameter and $\nu>0$ is the scale parameter. The failure level is normalized to one and divided into $K$ equal intervals, corresponding to condition states $0,\ldots,K-1$, with state $K$ representing failure. Following \cite{de2019discretizing}, the Gamma increment distribution is used to calculate the probabilities of remaining in the current state or transitioning to one or more worse states, producing the degradation transition matrix $P^X$. Thus, degradation is monotone and failure is absorbing in the absence of maintenance.

Rather than specifying $\kappa$ and $\nu$ directly, we parameterize degradation using the mean time to failure $\mu$ and its standard deviation $\sigma$. For each $(\mu,\sigma)$ pair, we numerically select $(\kappa,\nu)$ so that the absorption time of the resulting discrete Markov chain has the specified mean and standard deviation. This allows degradation speed and uncertainty to be varied using interpretable lifetime characteristics while systematically generating the corresponding transition matrix. This Gamma-process specification is used only for the numerical experiments. 

\subsubsection{Accessibility Process}
\label{sec:accessibility_experiments}

We represent accessibility using a two-state, seasonally varying Markov process, where the turbine is either accessible (\emph{A}) or inaccessible (\emph{I}) for maintenance. The process is characterized by three parameters: the persistence parameter $\rho$, which controls the tendency of the current accessibility state to persist into the following week; the seasonality amplitude $\Delta$, which determines the strength of seasonal variation in accessibility; and the peak-accessibility week $\bar{t}$, which identifies the time of year when accessibility conditions are most favorable. This specification captures both the temporal dependence of offshore weather \citep{hagen2013multivariate} and its seasonal variation \citep{martini2017accessibility}. For calendar week $t$, the seasonal adjustment is
\[
s_t=\Delta\cos\left(\frac{2\pi(t-\bar{t})}{52}\right),
\]

The probability that the turbine is accessible in week $t+1$ depends on its accessibility state in week $t$:
\[
q_t^I=P(A_{t+1}\mid I_t) =[1-\rho+s_t],
\]
\[
q_t^A=P(A_{t+1}\mid A_t)=[\rho+s_t],
\]
where
\[
[z]=\min\{0.99,\max\{0.01,z\}\}
\]
bounds each probability between $0.01$ and $0.99$.

The resulting transition matrix is
\[
P_t^W=
\begin{pmatrix}
1-q_t^I & q_t^I\\
1-q_t^A & q_t^A
\end{pmatrix},
\]
where rows and columns correspond to $(I,A)$. When $s_t=0$, both accessibility states have a persistence probability of $\rho$. A positive seasonal adjustment increases the probability of accessibility in the following week, while a negative adjustment decreases it. The complementary probabilities ensure that each row sums to one, including when the probability bounds are applied.

\subsubsection{Parameterization}\label{sec:param}

Table~\ref{tab:base_case_parameters} summarizes the parameter values used in the base-case experiments. The maintenance-cost parameters are selected to reflect the well-established cost ordering between preventive and corrective interventions \citep{panagiotidou2010statistical}. Corrective maintenance after turbine failure generally requires greater repair resources, longer downtime, and potentially more expensive logistics than preventive intervention \citep{raza2019optimal,dao2021integrated}. Accordingly, we set the preventive and corrective maintenance costs to $c^{\mathrm{PM}}=\$150{,}000$ and $c^{\mathrm{CM}}=\$750{,}000$, respectively. These values are intended to preserve the economically relevant ordering and relative magnitude of the two interventions rather than reproduce a specific turbine's monetary costs.

\begin{table}[htbp]
\centering
\caption{Base-case parameter values}
\label{tab:base_case_parameters}
\renewcommand{\arraystretch}{1.2}
\begin{tabular}{lll}
\toprule
Parameter & Symbol & Value \\
\midrule
Number of condition states 
& $|\mathcal{X}|$ 
& 11 \\

Number of weeks in seasonal cycle 
& $T$ 
& 52 weeks\\

Preventive maintenance cost 
& $c^{\mathrm{PM}}$ 
& \$150,000 \\

Corrective maintenance cost 
& $c^{\mathrm{CM}}$ 
& \$750,000 \\

Mean time to failure 
& $\mu$ 
& 80 weeks\\

Standard deviation of time to failure 
& $\sigma$ 
& 35 weeks\\

Weather persistence 
& $\rho$ 
& 0.55 \\

Seasonality amplitude 
& $\Delta$ 
& 0.40 \\

Peak-accessibility week 
& $\bar{t}$ 
& 30 \\
\bottomrule
\end{tabular}

\vspace{1mm}
\begin{minipage}{0.95\linewidth}
\centering
\footnotesize
\textit{Note:} Monetary parameters are expressed in USD, and time parameters are expressed in weeks.
\end{minipage}

\end{table}

Continued operation in a degraded condition is associated with a state-dependent efficiency loss. Such degradation-dependent operating costs have been explicitly considered in CBM models, with operating costs increasing as the system degrades \citep{panagiotidou2010statistical,liu2017condition}, and consistent with empirical evidence that wind-asset performance and energy production decline with aging and degradation \citep{staffell2014does}. Accordingly, we use an increasing operating-cost profile $g(x)$, with a substantially larger cost assigned to the failed state $g(K)$ (see Figure \ref{fig:eff_loss}) to capture the economic consequences of turbine unavailability and lost production. 

\begin{figure}[htbp]
\includegraphics[width=0.5\columnwidth]{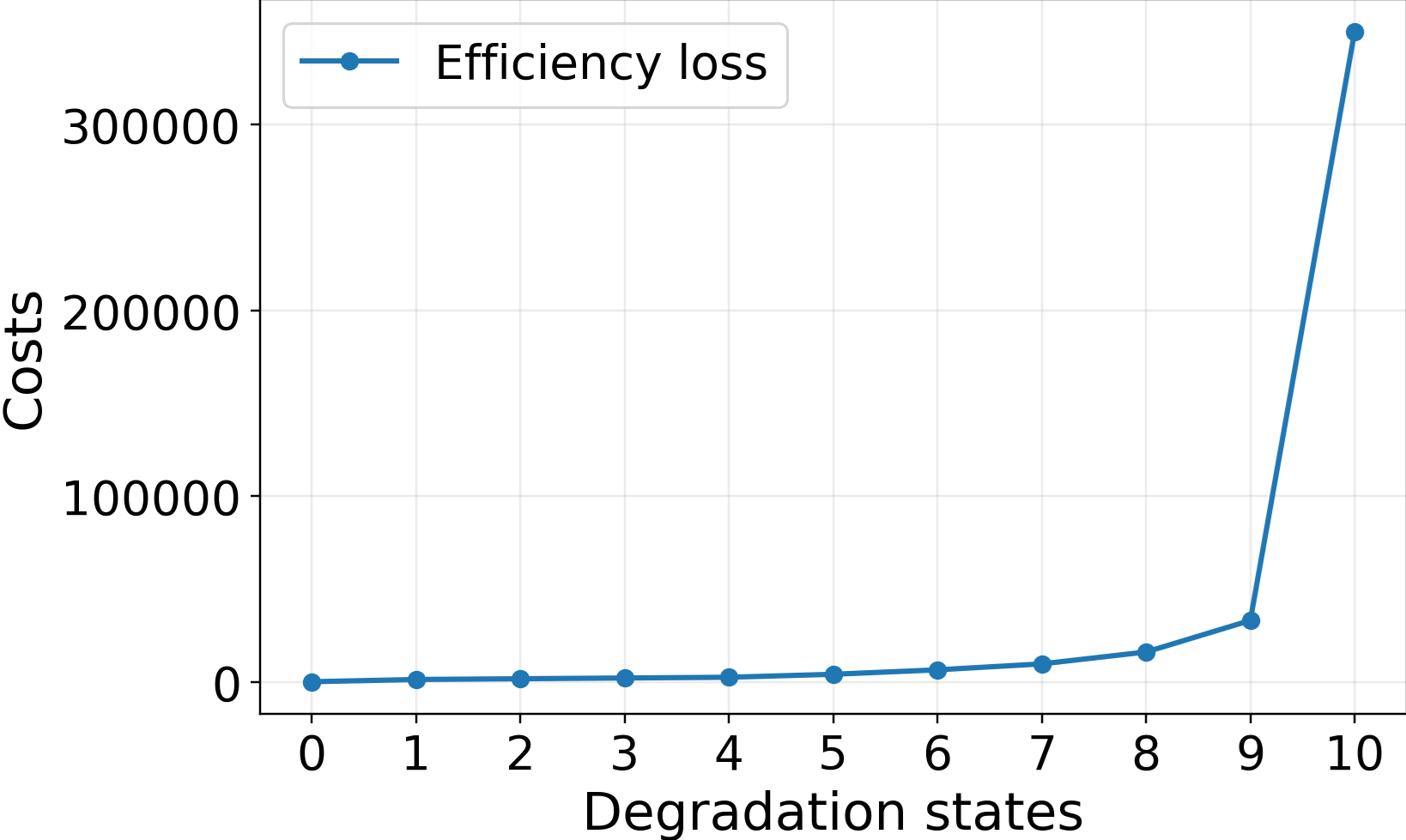}
\centering
\caption{Operating-cost profile.}
\label{fig:eff_loss}
\end{figure} 

The degradation-process parameters are selected to generate a monotone and stochastic degradation pattern consistent with the gradual loss of turbine condition observed in offshore wind applications. The asset is represented by 11 condition states. We set the mean time to failure to $\mu=80$ weeks, consistent with reported wind-turbine failure timescales \citep{li2022failure}, and the standard deviation to $\sigma=35$ weeks to capture the substantial uncertainty in degradation timing observed in wind-turbine reliability data \citep{dao2019wind,kiadaliry2026weibull}. 

The accessibility-process parameters are selected to generate a persistent \citep{hagen2013multivariate} and seasonally varying accessibility pattern \citep{martini2017accessibility} consistent with empirical offshore weather characteristics.  Moreover, North Sea data indicate substantially more favorable weather conditions during the summer months \citep{hagen2013multivariate, martini2017accessibility}. Accordingly, the base case follows a 52-week seasonal cycle with a moderate persistence level $\rho=0.55$, a pronounced seasonality amplitude $\Delta=0.40$, and peak accessibility in week $\bar{t}=30$, corresponding to late July. 

\subsection{Base-Case Analysis}
\label{subsec:base_case}

We first examine the optimal maintenance policy for the base case. In particular, we investigate how the optimal maintenance threshold varies over the annual accessibility cycle and how these variations relate to both current accessibility and the expected evolution of future maintenance opportunities. Figure~\ref{fig:base_case} presents the probabilities of accessibility in the following week, conditional on the current accessibility state, together with the optimal maintenance threshold. The threshold indicates the condition state at which maintenance becomes optimal when the turbine is accessible. 
The figure shows that the maintenance threshold changes systematically over the seasonal cycle. The optimal threshold is generally higher when current accessibility is favorable and when favorable accessibility is expected to persist or improve in the near future. In these periods, the operator is more willing to defer maintenance and continue operating the turbine in a more degraded state. By contrast, when future accessibility is expected to worsen, the threshold is lower, meaning that the operator is more willing to use the current maintenance opportunity before accessibility  becomes more limited.

\begin{figure}[htbp]
\includegraphics[width=1\columnwidth]{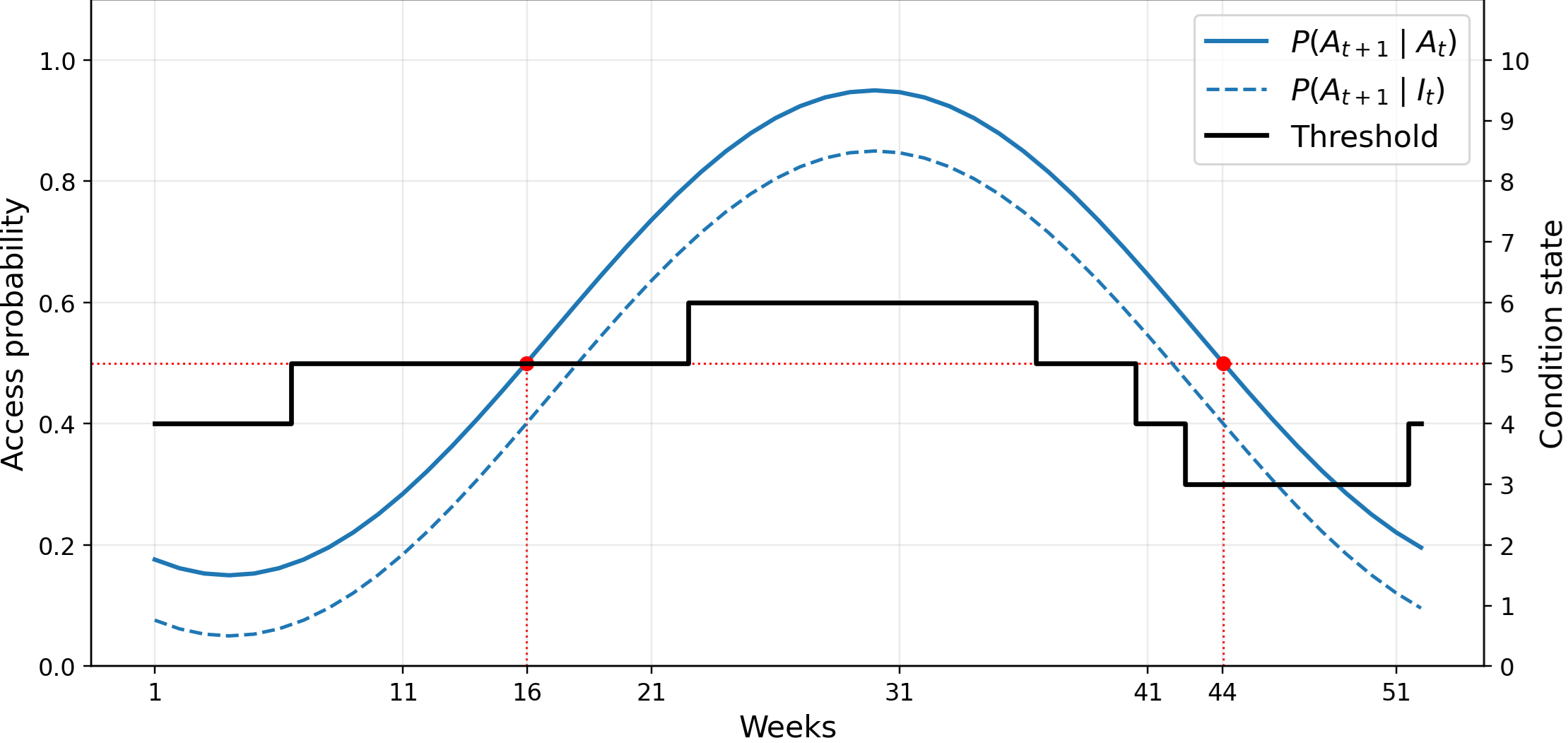}
\centering
\caption{Optimal maintenance thresholds over the seasonal cycle with the weekly accessibility transition probabilities.}
\label{fig:base_case}
\end{figure} 

The optimal threshold depends not only on next-week's accessibility probability but also on the weeks that follow. For instance, weeks 16 and 44 have approximately the same one-week accessibility probabilities (0.50), yet their optimal thresholds are 5 and 3, respectively. Accessibility is expected to improve after week 16 but worsen after week 44. Consequently, maintenance can be postponed to a more degraded condition in week 16, whereas the worsening accessibility outlook favors earlier intervention in week 44.


This comparison highlights the main operational insight of the model: optimal CBM decisions are shaped not only by the current condition of the turbine and current accessibility, but also by the expected evolution of future maintenance opportunities. Stochastic accessibility therefore preserves the familiar threshold structure of CBM, but makes the threshold state-dependent and forward-looking.

\subsection{Value of Informed Decisions}
\label{subsec:value_information}

We next evaluate the value of incorporating different types of information into maintenance decisions. The optimal policy uses both the condition of the asset and the accessibility process. 
To isolate the contribution of these information sources, we compare it with two benchmark policies that use progressively more information.

The first benchmark is an optimized age-based maintenance policy. This policy does not use the degradation state to determine the timing of preventive maintenance. Instead, preventive maintenance becomes due once a specified number of weeks has elapsed since the most recent maintenance intervention. Comparing the optimal policy with the age-based policy captures the combined value of using asset condition and accessibility information. 
We evaluate maintenance ages from 4 to 44 weeks and select the age that minimizes the long-run average cost. As shown in Figure~\ref{fig:benchmark_optimization}(a), the minimum is attained at a maintenance age of 24 weeks.

The second benchmark is an optimized constant-threshold policy. This policy observes the current turbine condition and performs preventive maintenance once the condition reaches a single threshold that is applied throughout the year. Thus, it retains condition-based decision making but does not adapt the threshold to the accessibility process. 
The difference between this benchmark and the optimal policy therefore isolates the value of incorporating accessibility information into CBM decisions. 
We evaluate every feasible constant threshold and select the one that minimizes the long-run average cost. Figure~\ref{fig:benchmark_optimization}(b) reports the resulting long-run average cost for each threshold. The minimum is attained at a condition threshold of 4, which is therefore used as the optimized constant-threshold benchmark.

\begin{figure}[htbp]
\centering

\begin{minipage}[t]{0.48\textwidth}
    \centering
    \includegraphics[width=\linewidth]{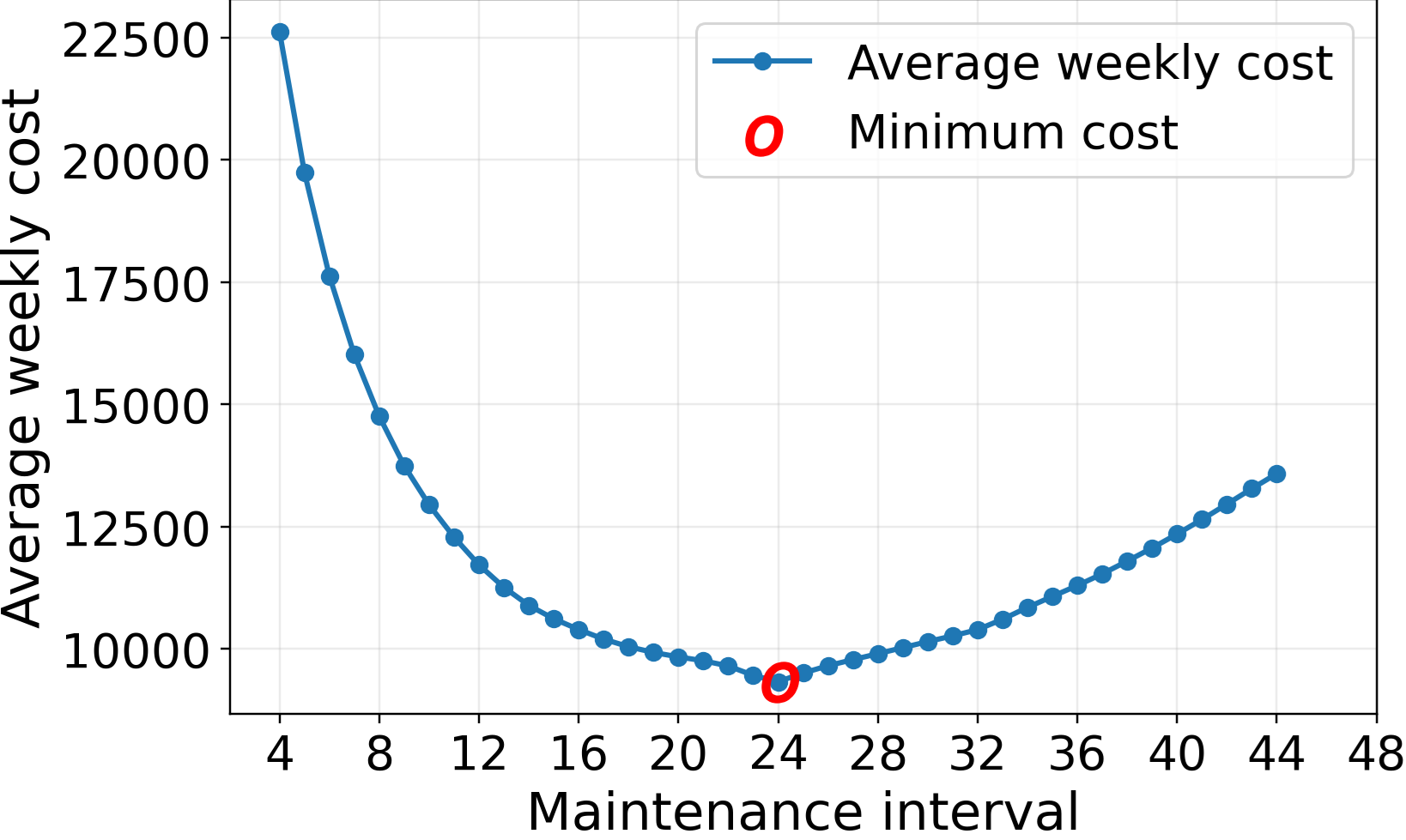}
    \small (a) Age-based maintenance policy
\end{minipage}\hfill
\begin{minipage}[t]{0.48\textwidth}
    \centering
    \includegraphics[width=\linewidth]{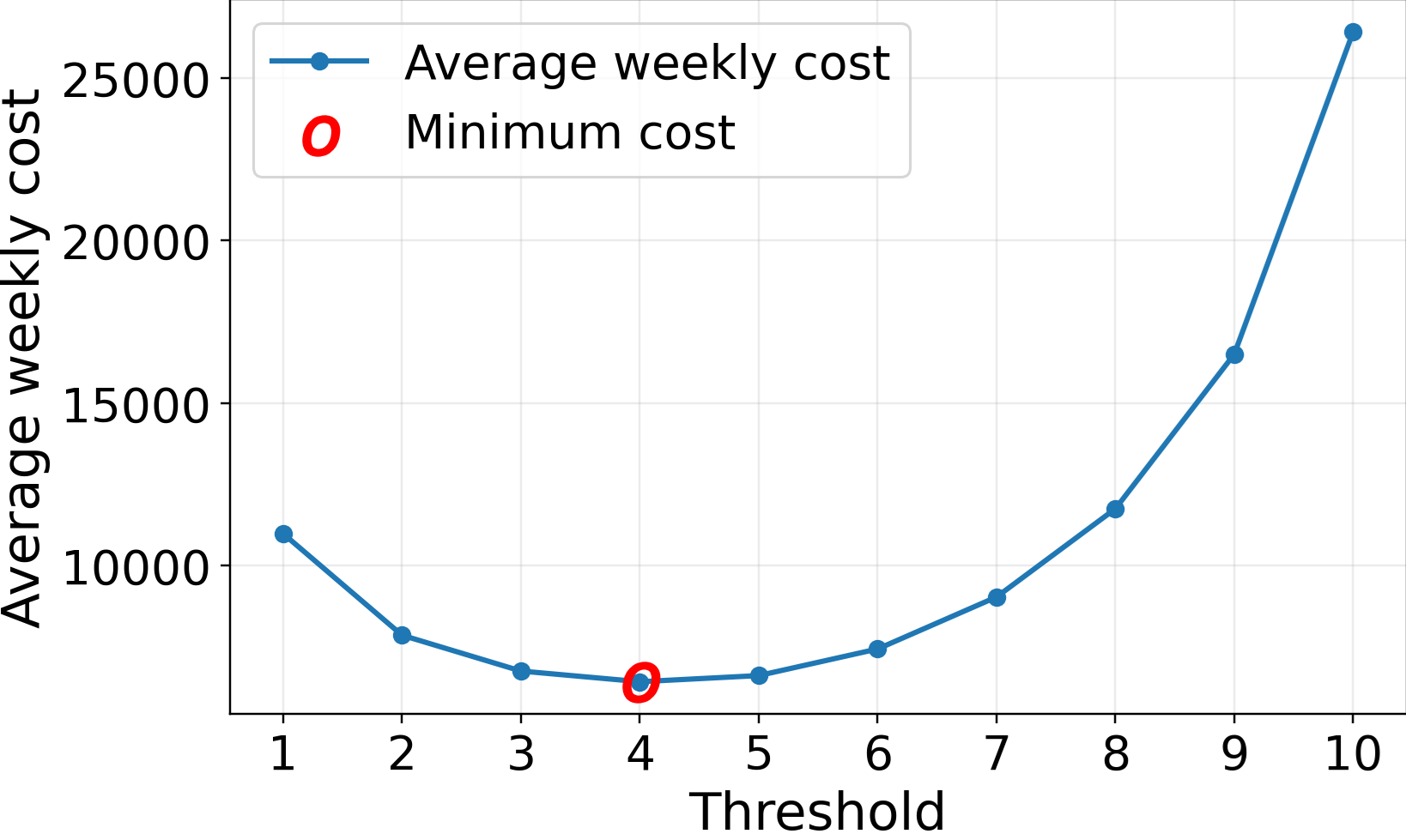}
    \small (b) Constant-threshold policy
\end{minipage}

\caption{Panel (a) reports the long run average cost for maintenance ages between 4 and 44 weeks. Panel (b) reports the long run average cost for each feasible constant condition threshold.}
\label{fig:benchmark_optimization}
\end{figure}

We next compare the long-run performance of the three policies. Table~\ref{tab:value_information} reports seven performance measures. Maintenance frequency is the average total number of interventions per year, while PM and CM frequencies report the average annual numbers of preventive and corrective interventions, respectively. The operating cost rate measures the average weekly cost of efficiency losses and failure-related downtime. The PM and CM cost rates represent the average weekly expenditures on preventive and corrective maintenance, respectively. Finally, the long-run average cost is the sum of these three cost rates and represents the overall economic performance of each policy. The percentages in parentheses indicate the contribution of each cost component to the total.
Together, these metrics allow us to see not only which policy is less costly, but also why the cost differences arise.

\begin{table}[htbp]
\centering
\caption{Performance comparison under the base-case setting}
\label{tab:value_information}
\renewcommand{\arraystretch}{1.2}
\begin{tabular}{lrrr}
\toprule
Metric
& Age-based
& Constant threshold
& Optimal policy \\
\midrule

Maintenance frequency 
& 1.96 & 1.35 & 1.30 \\

PM frequency 
& 1.88 & 1.31 & 1.26 \\

CM frequency 
& 0.075 & 0.036 & 0.038 \\

Operating cost rate
& \$2816.1 (30.21\%)
& \$2127.3 (33.06\%)
& \$1983.0 (32.20\%) \\

PM cost rate
& \$5429.8 (58.24\%)
& \$3783.2 (58.79\%)
& \$3629.6 (58.94\%) \\

CM cost rate
& \$1077.1 (11.55\%)
& \$524.5 (8.15\%)
& \$545.1 (8.85\%) \\

Long-run average cost
& \textbf{\$9323.0}
& \textbf{\$6435.0}
& \textbf{\$6157.7} \\
\bottomrule
\end{tabular}

\vspace{1mm}
\begin{minipage}{0.95\linewidth}
\centering
\footnotesize
\textit{Note:} Maintenance frequencies are reported as the expected number of events per year.
\end{minipage}

\end{table}

The difference between the optimal policy and the age-based policy showcases the value of incorporating asset condition and accessibility information into maintenance decisions. The optimal policy reduces long-run average cost from \$9323.0 to \$6157.7, corresponding to a 33.95\% improvement. It also performs fewer preventive interventions and experiences roughly half as many corrective interventions, with CM frequency decreasing from 0.075 to 0.038 events per year. These results demonstrate the substantial value of considering both the asset’s condition state and accessibility state when making maintenance decisions, rather than relying solely on elapsed time since the previous intervention. 

The difference between the optimal policy and the constant-threshold policy showcases the additional value of adapting maintenance decisions to changing accessibility conditions. Relative to the constant-threshold policy, the optimal policy reduces long-run average cost from \$6435.0 to \$6157.7, a 4.31\% improvement. Both policies observe the degradation state and are subject to the same accessibility restrictions; the difference is that the optimal policy adjusts its threshold according to the current accessibility state and the evolution of future maintenance opportunities. The gain comes from lower operating and preventive maintenance cost rates, which more than offset a small increase in corrective-maintenance cost.

Taken together, these comparisons quantify the economic value of the two information sources. Incorporating condition information reduces long-run average cost from \$9323.0 to \$6435.0, while adapting the maintenance threshold to accessibility provides a further reduction to \$6157.7. The results show that both condition information and accessibility dynamics are economically important.


\begin{figure}[htbp]
\centering

\begin{minipage}{0.9\linewidth}
    \centering
    \includegraphics[width=\linewidth]{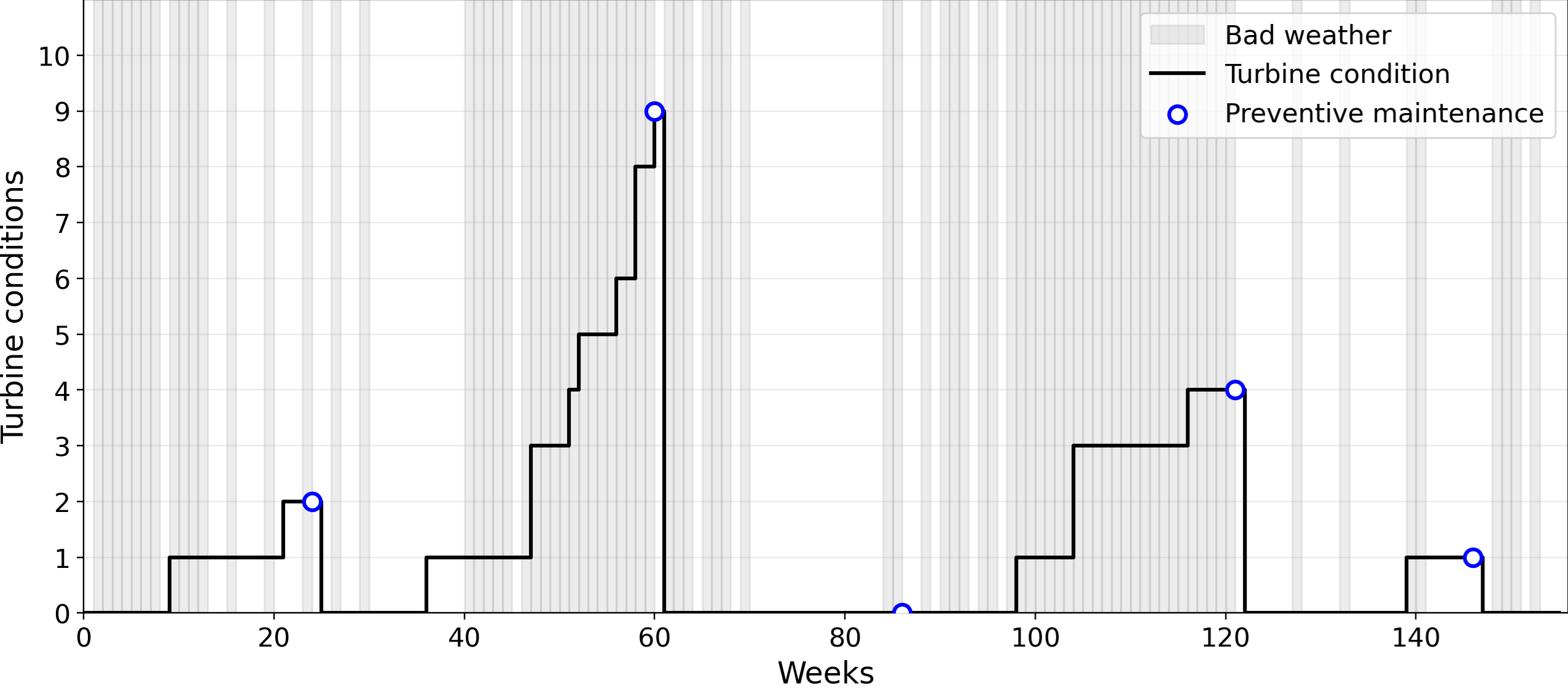}
    \small (a) Age-based maintenance policy
\end{minipage}

\vspace{2mm}

\begin{minipage}{0.9\linewidth}
    \centering
    \includegraphics[width=\linewidth]{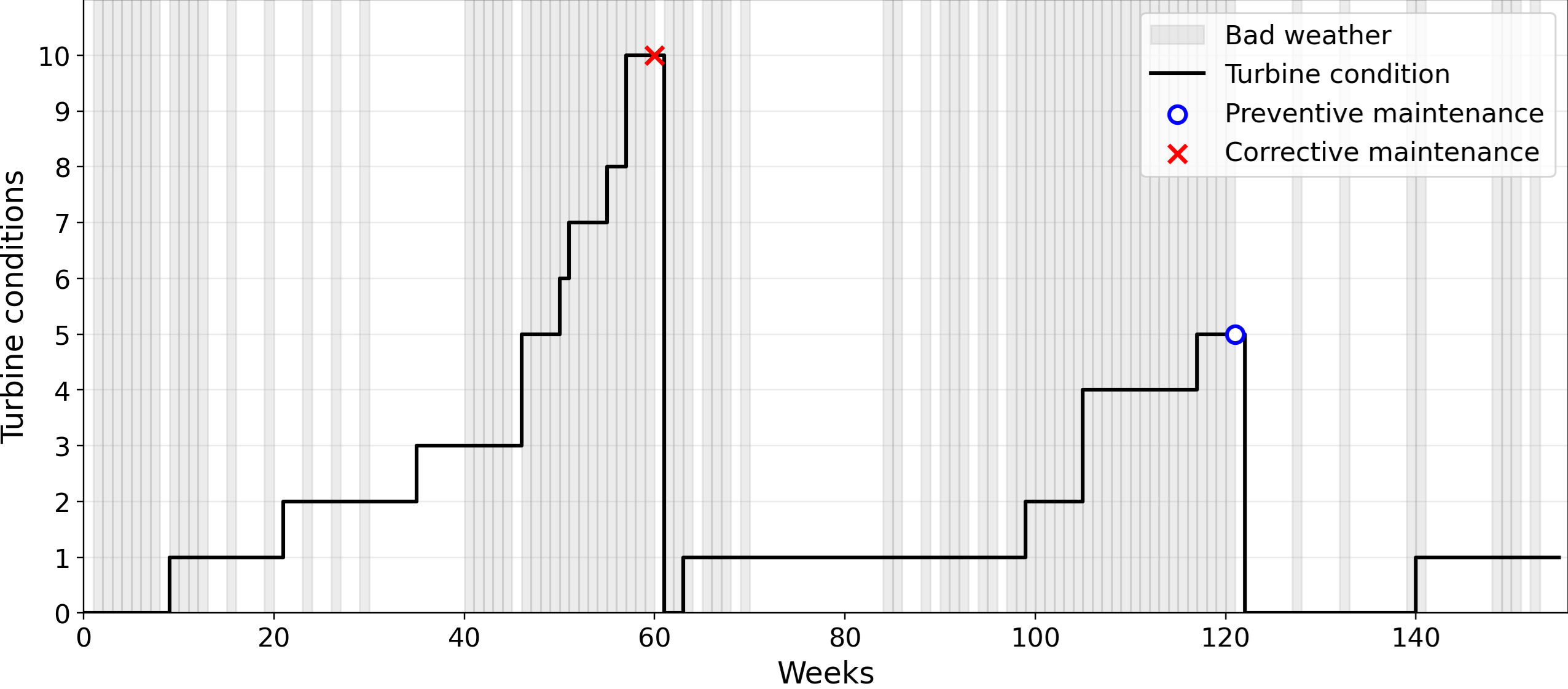}
    \small (b) Constant-threshold policy
\end{minipage}

\vspace{2mm}

\begin{minipage}{0.9\linewidth}
    \centering
    \includegraphics[width=\linewidth]{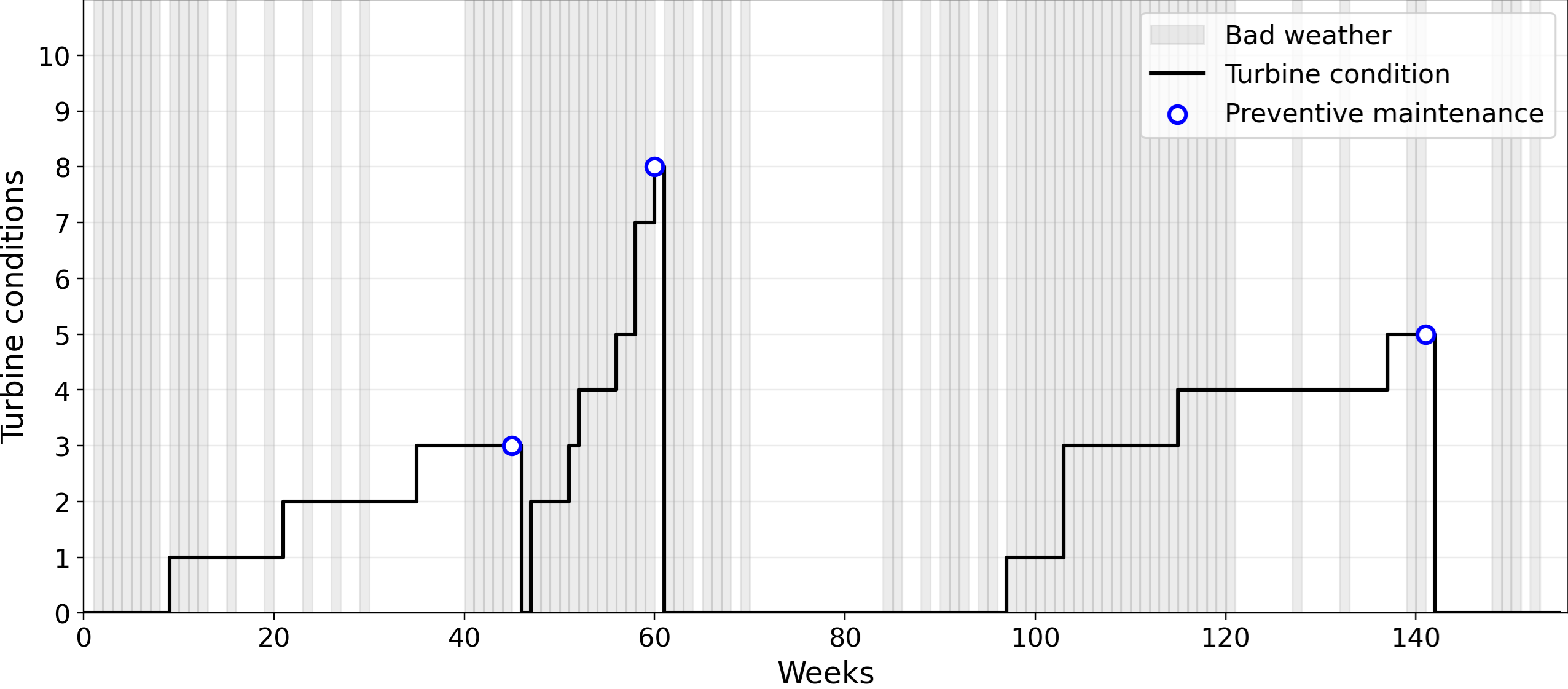}
    \small (c) Optimal policy
\end{minipage}

\caption{Illustrative simulated condition trajectories under the three maintenance policies.}
\label{fig:policy_simulation_comparison}
\end{figure}

Figure~\ref{fig:policy_simulation_comparison} illustrates the maintenance decisions under (a) the age-based policy, (b) the constant-threshold policy, and (c) the optimal policy. Each panel shows the turbine's condition, accessible and inaccessible weeks, and preventive and corrective maintenance actions over a 150-week simulation. The purpose of this comparison is to show how the different information used by each policy translates into different maintenance decisions along the same operating trajectory. To ensure a direct comparison, the same realized accessibility conditions are used for all three simulations, and the same random seed is used for the stochastic degradation process. Thus, differences in the resulting condition trajectories and maintenance events arise from the decision rules themselves rather than from different realizations of accessibility or degradation.


Figure~\ref{fig:policy_simulation_comparison}(a) shows that the age-based policy performs preventive maintenance while the turbine is still in the new condition state, illustrating how reliance on elapsed time can lead to unnecessary interventions. A different effect appears in Figure~\ref{fig:policy_simulation_comparison}(b), in which the constant-threshold policy skips an available maintenance opportunity in week 46, because the turbine has not reached its fixed threshold of 4. In contrast, the optimal policy in Figure~\ref{fig:policy_simulation_comparison}(c) lowers its threshold to 3 during weeks 43--51 and performs preventive maintenance in week 46, before the extended period of inaccessibility. With no further accessibility until week 60, the turbine operating under the constant-threshold policy subsequently fails. These examples illustrate how the optimal policy uses condition information to avoid unnecessary maintenance and the future accessibility outlook to determine when an available maintenance opportunity should be used.

Taken together, these trajectories demonstrate the distinct roles of condition and accessibility information. Relying solely on elapsed time can lead to unnecessary preventive maintenance, while failing to adapt maintenance thresholds to the accessibility outlook can result in missed maintenance opportunities and subsequent failure during prolonged periods of inaccessibility. 

\subsection{Sensitivity Analysis}

We next examine how the structure of the optimal maintenance policy changes with the main characteristics of the accessibility process, degradation process, and costs. For each experiment, one parameter is varied around the base-case value while all remaining parameters are fixed, and the optimal policy is recomputed. Table~\ref{tab:sensitivity_values} summarizes the parameter values considered in the sensitivity analysis, organized into three groups: accessibility characteristics, comprising weather persistence ($\rho$) and seasonality amplitude ($\Delta$); degradation characteristics, comprising the mean ($\mu$) and standard deviation ($\sigma$) of time to failure; and economic characteristics, comprising preventive maintenance cost ($c^{PM}$) and the efficiency-loss multiplier ($\gamma_g$).

\begin{table}[htbp]
\centering
\caption{Sensitivity-analysis parameter values}
\label{tab:sensitivity_values}
\renewcommand{\arraystretch}{1.2}
\begin{tabular}{lllccc}
\toprule
Category & Factor & Symbol & \multicolumn{3}{c}{Values considered} \\
\cmidrule(lr){4-6}
& & & Low & Base & High \\
\midrule

\multirow{2}{*}{Economic}
& Preventive maintenance cost
& $c^{\mathrm{PM}}$
& \$25{,}000
& \$150{,}000
& \$749{,}999 \\

& Efficiency-loss multiplier
& $\gamma_g$
& 0.25
& 1.00
& 4.00 \\

\midrule

\multirow{2}{*}{Degradation}
& Mean time to failure
& $\mu$
& 60 weeks
& 80 weeks
& 100 weeks \\

& Standard deviation of time to failure
& $\sigma$
& 25 weeks
& 35 weeks
& 45 weeks \\

\midrule

\multirow{2}{*}{Accessibility}
& Weather persistence
& $\rho$
& 0.50
& 0.55
& 0.60 \\

& Seasonality amplitude
& $\Delta$
& 0.30
& 0.40
& 0.45 \\

\bottomrule
\end{tabular}
\end{table}

For each parameter, we present the optimal maintenance thresholds over the 52-week cycle under the low, base, and high settings. The figures illustrate how changes in each parameter affect both the overall threshold level and its seasonal variation, allowing us to identify when maintenance should be performed earlier or postponed in response to different system characteristics.

\subsubsection{Accessibility Characteristics}

Weather persistence determines how long favorable or unfavorable accessibility conditions are likely to last, thereby influencing the risk of postponing a maintenance. Figure~\ref{fig:accessibility_sensitivity} shows how the optimal thresholds change with the weather-persistence parameter $\rho$. As persistence increases, thresholds become slightly lower during unfavorable periods and slightly higher during favorable periods, because current accessibility becomes more informative about the conditions likely to follow. Thus, stronger persistence makes the policy more responsive to the accessibility outlook: maintenance is performed earlier when poor accessibility is likely to persist and can be postponed when favorable conditions is likely to continue.

\begin{figure}[htbp]
\centering

\begin{minipage}[t]{0.48\linewidth}
    \centering
    \includegraphics[width=\linewidth]{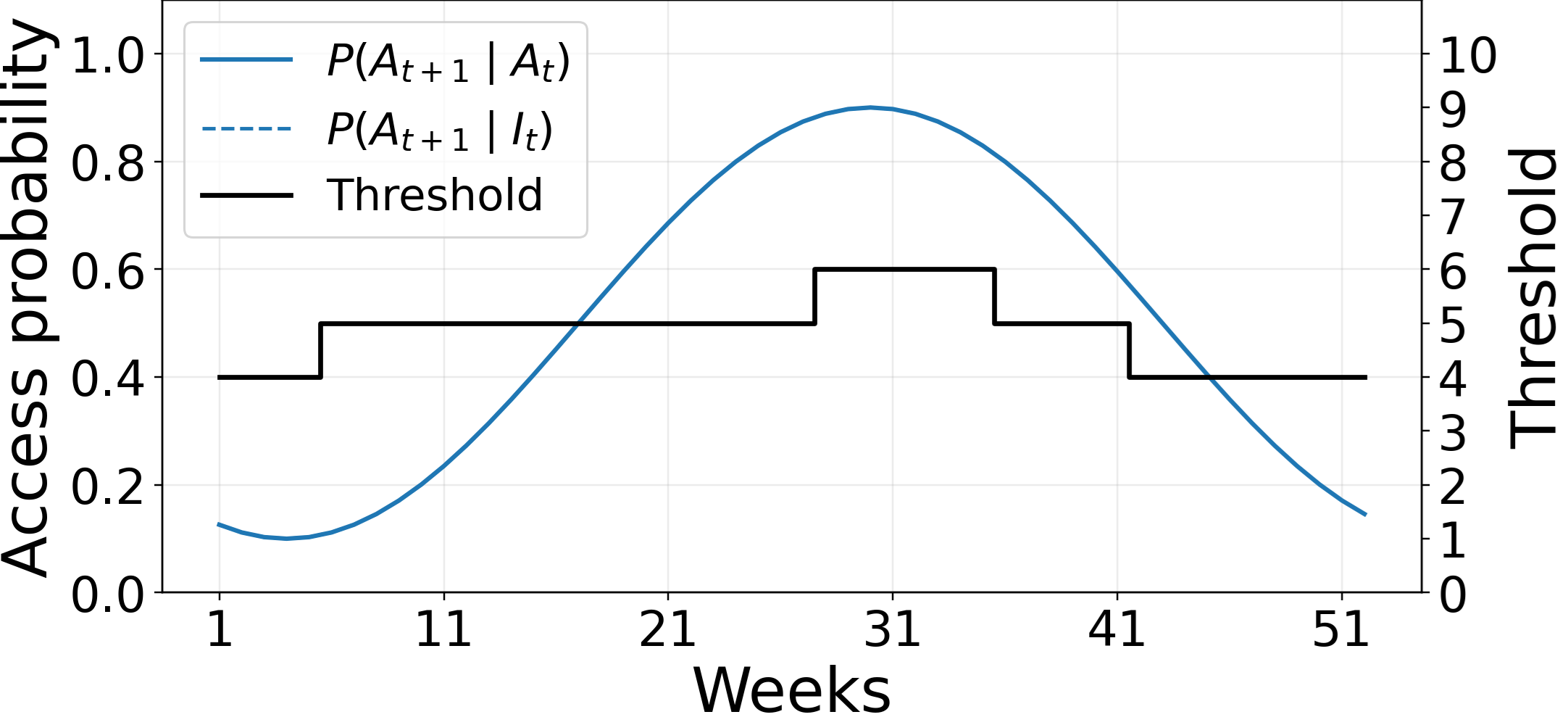}
    \small (a) $\rho=0.50, \ \Delta=0.40$
\end{minipage}
\hfill
\begin{minipage}[t]{0.48\linewidth}
    \centering
    \includegraphics[width=\linewidth]{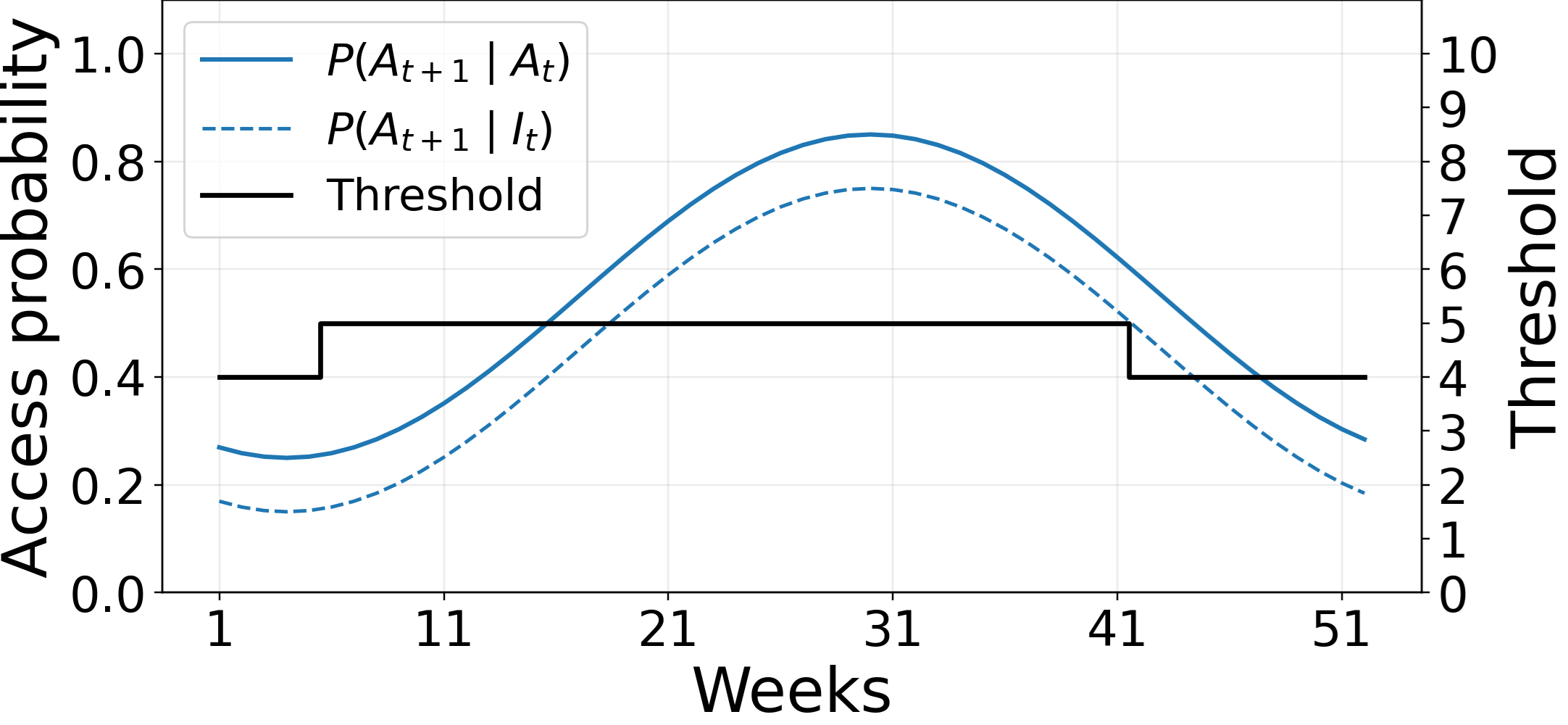}
    \small (b) $\rho=0.55, \ \Delta=0.30$
\end{minipage}

\vspace{2mm}

\begin{minipage}[t]{0.48\linewidth}
    \centering
    \includegraphics[width=\linewidth]{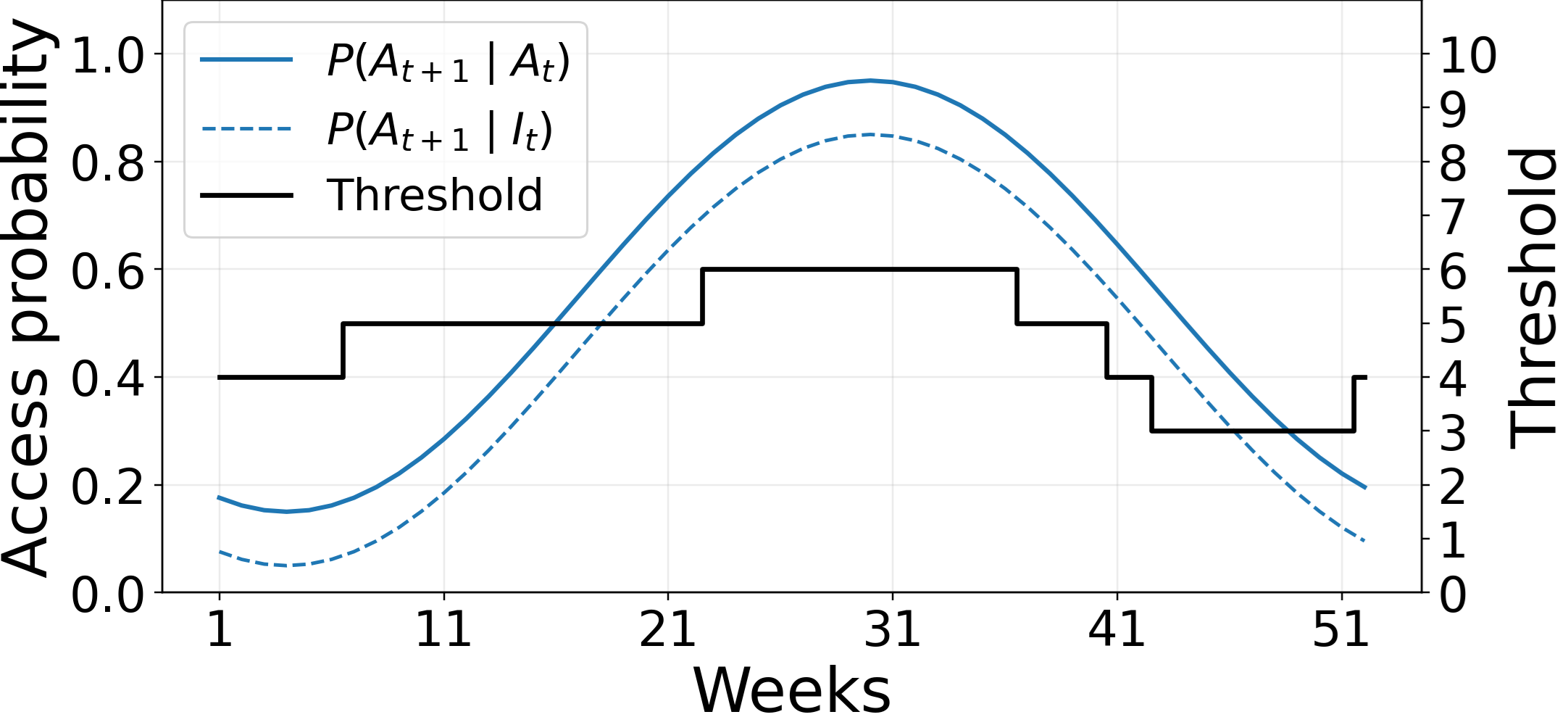}
    \small (c) $\rho=0.55, \ \Delta=0.40$
\end{minipage}
\hfill
\begin{minipage}[t]{0.48\linewidth}
    \centering
    \includegraphics[width=\linewidth]{figures/base_case.png}
    \small (d) $\rho=0.55, \ \Delta=0.40$
\end{minipage}

\vspace{2mm}

\begin{minipage}[t]{0.48\linewidth}
    \centering
    \includegraphics[width=\linewidth]{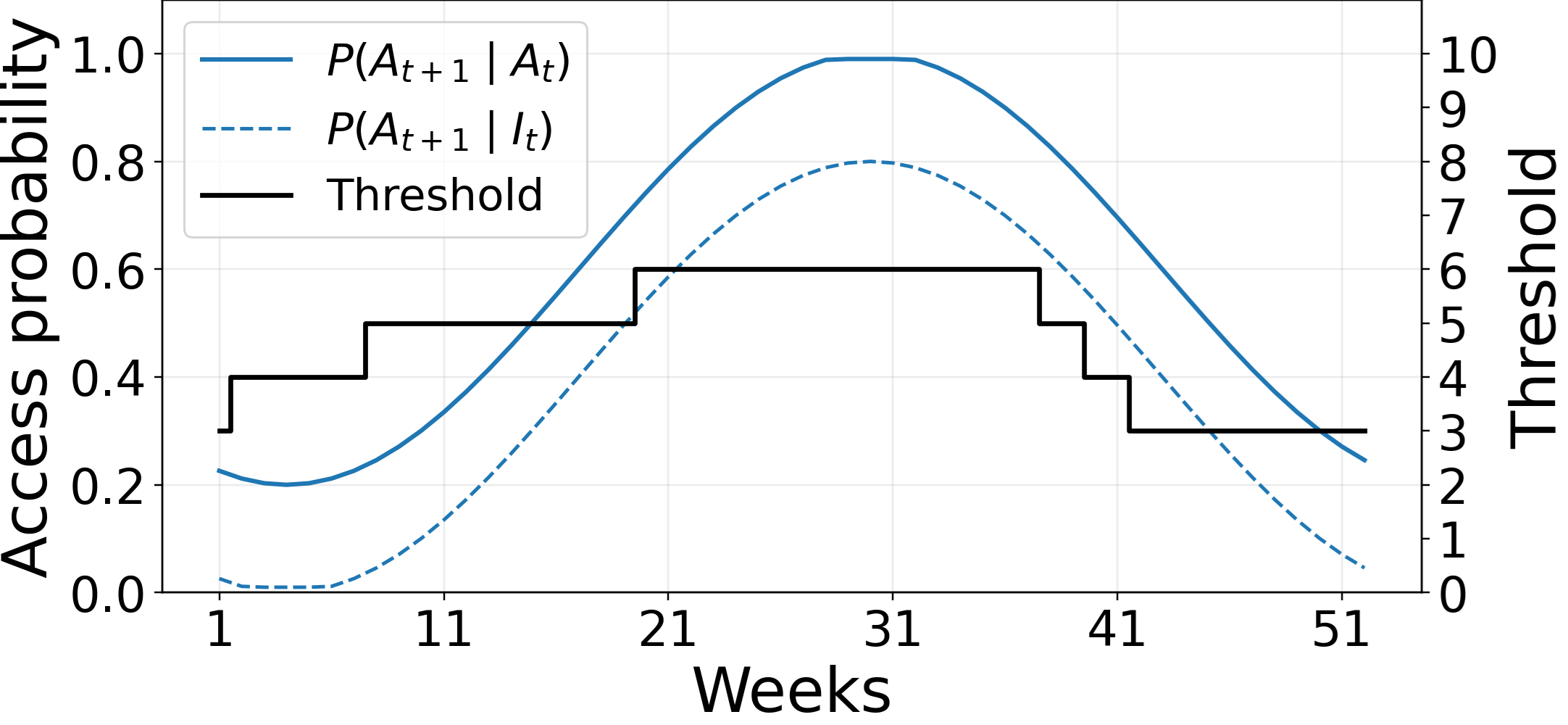}
    \small (e) $\rho=0.60, \ \Delta=0.40$
\end{minipage}
\hfill
\begin{minipage}[t]{0.48\linewidth}
    \centering
    \includegraphics[width=\linewidth]{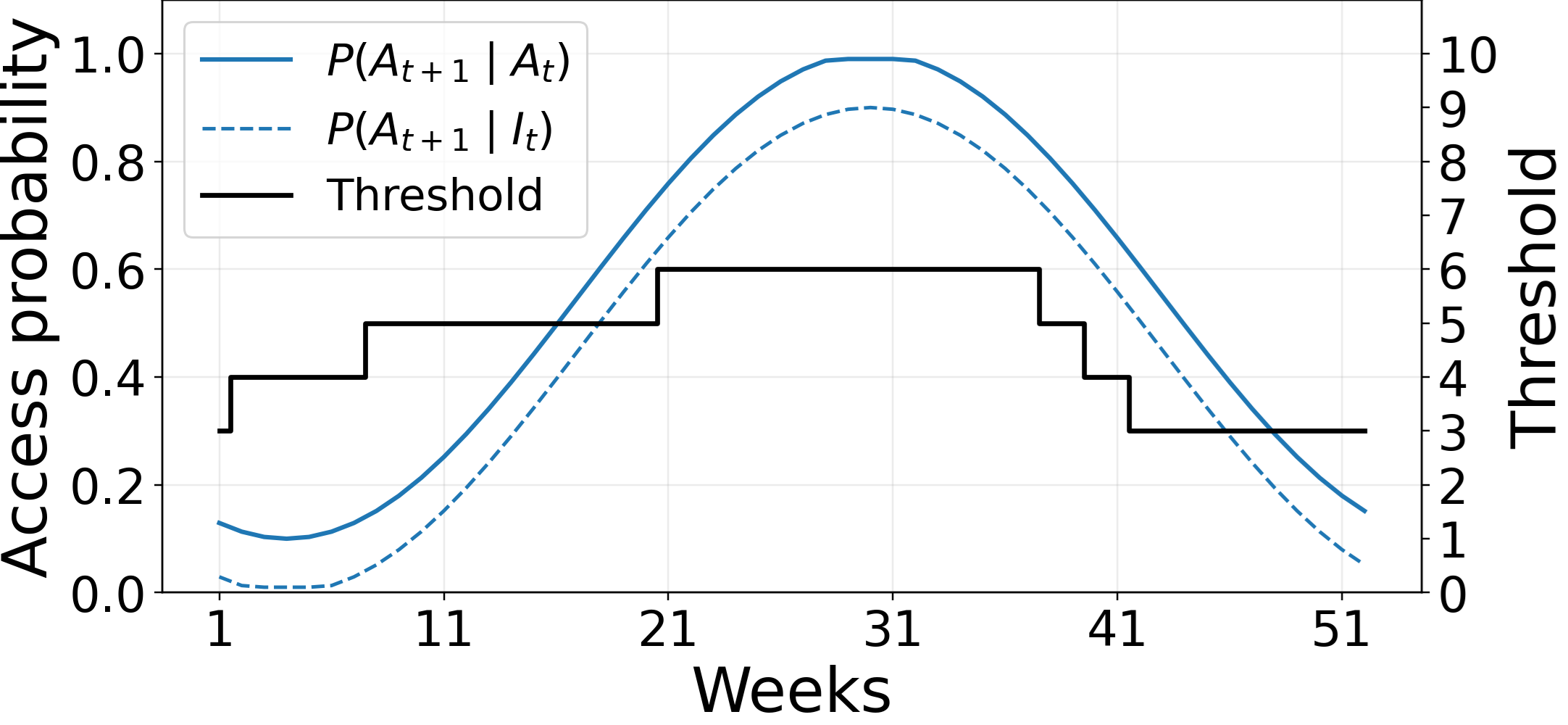}
    \small (f) $\rho=0.55, \ \Delta=0.45$
\end{minipage}

\caption{Sensitivity of the optimal maintenance thresholds to accessibility characteristics. The left column varies weather persistence $\rho$, while the right column varies seasonality amplitude $\Delta$.}
\label{fig:accessibility_sensitivity}
\end{figure}

Seasonality determines how maintenance opportunities are distributed throughout the year, influencing the importance of scheduling maintenance relative to favorable and unfavorable seasons. Figure~\ref{fig:accessibility_sensitivity} shows how the optimal thresholds change with the seasonality amplitude $\Delta$. When seasonality is weak, such as $\Delta=0.30$, the optimal threshold is nearly constant throughout the year. As seasonality strengthens, thresholds become higher during favorable periods and lower during unfavorable periods. Thus, stronger seasonality increases the need to adapt maintenance decisions to changing accessibility, while weak seasonality makes a constant threshold nearly sufficient.

\subsubsection{Degradation Characteristics}

The mean time to failure determines the pace of degradation and, consequently, how long maintenance can be postponed without excessive failure risk. Figure~\ref{fig:degradation_sensitivity} shows that shorter mean time to failure leads to lower optimal thresholds, while longer mean time to failure leads to higher thresholds. Thus, faster degradation calls for earlier intervention, leaving less flexibility to wait for future maintenance opportunities. 

\begin{figure}[htbp]
\centering

\begin{minipage}[t]{0.48\linewidth}
    \centering
    \includegraphics[width=\linewidth]{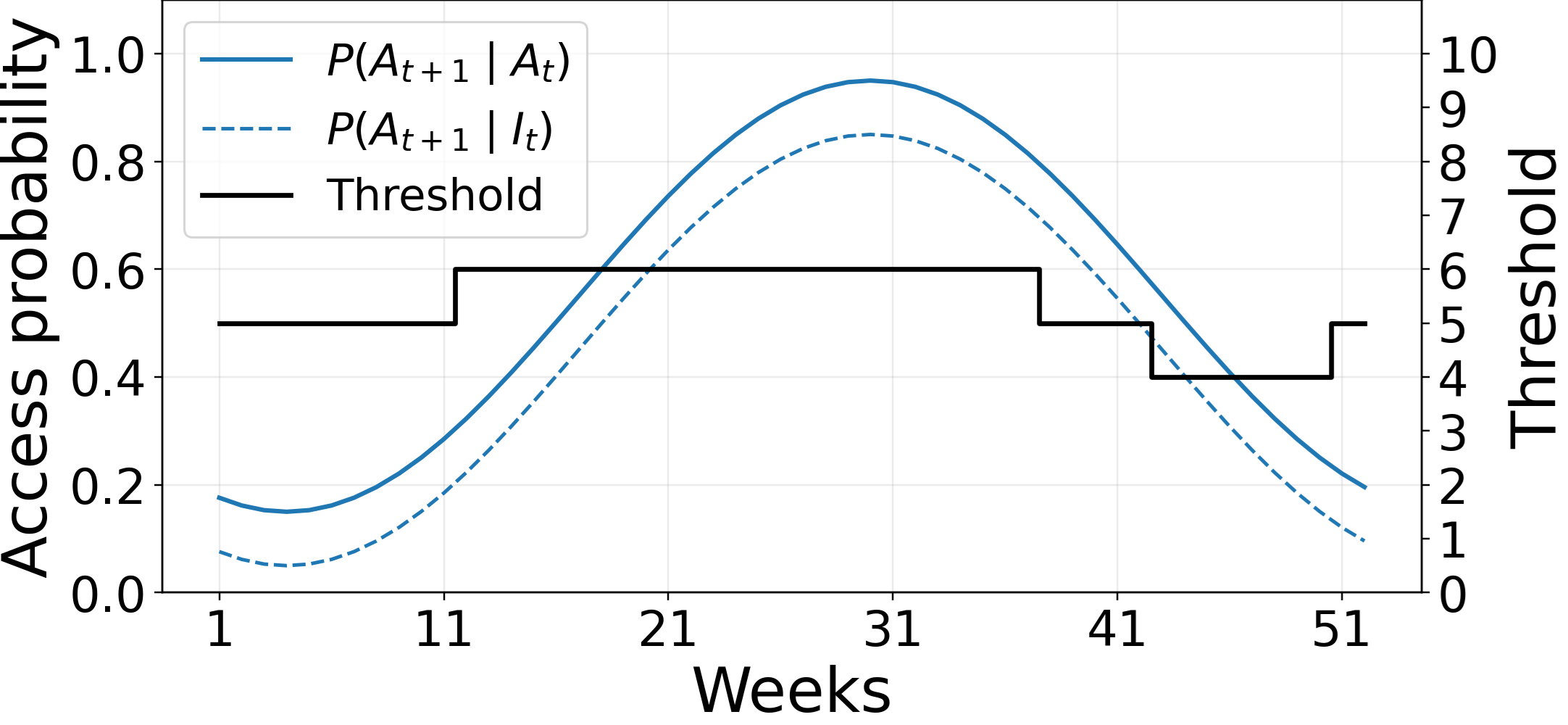}
    \small (a) $\mu=100, \ \sigma=35$
\end{minipage}
\hfill
\begin{minipage}[t]{0.48\linewidth}
    \centering
    \includegraphics[width=\linewidth]{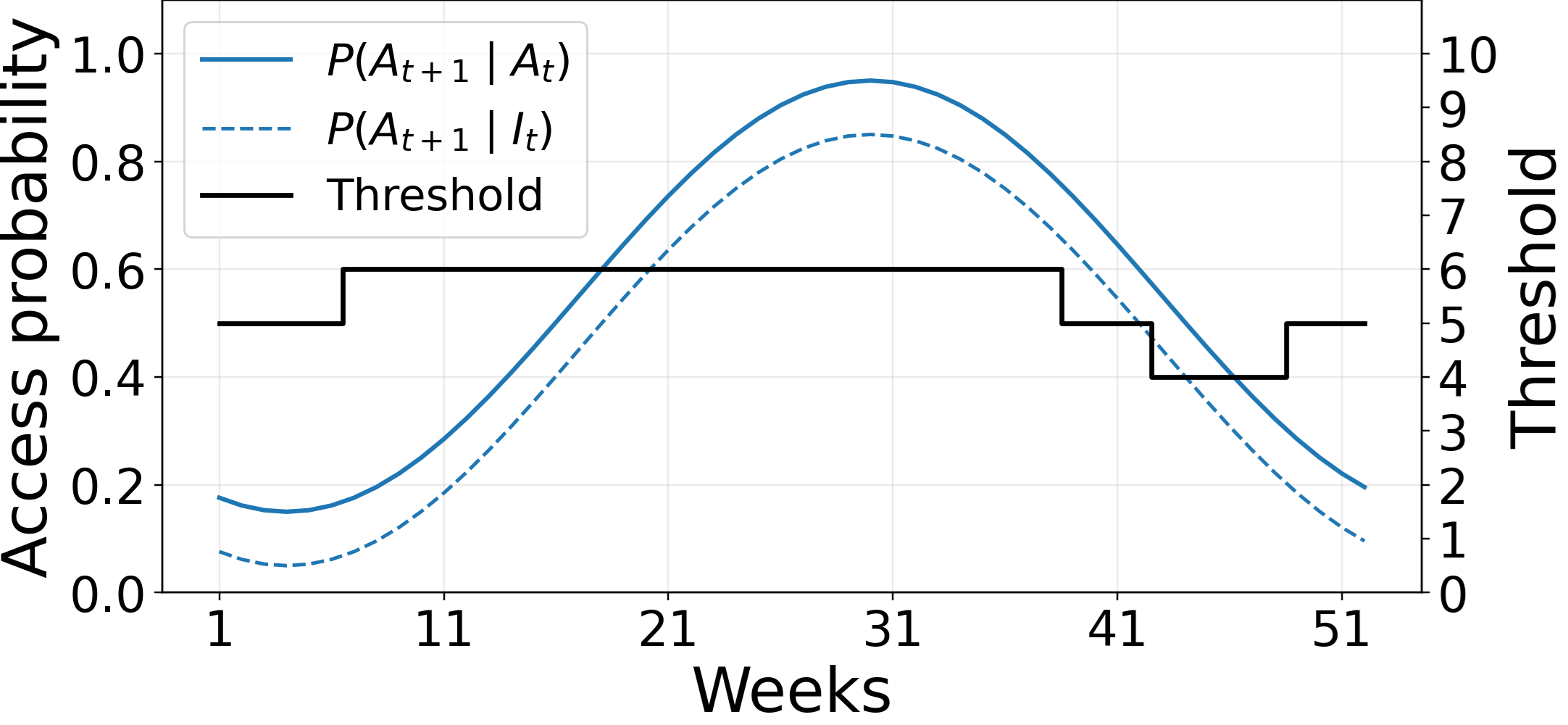}
    \small (b) $\mu=80, \ \sigma=25$
\end{minipage}

\vspace{2mm}

\begin{minipage}[t]{0.48\linewidth}
    \centering
    \includegraphics[width=\linewidth]{figures/base_case.png}
    \small (c) $\mu=80, \ \sigma=35$
\end{minipage}
\hfill
\begin{minipage}[t]{0.48\linewidth}
    \centering
    \includegraphics[width=\linewidth]{figures/base_case.png}
    \small (d) $\mu=80, \ \sigma=35$
\end{minipage}

\vspace{2mm}

\begin{minipage}[t]{0.48\linewidth}
    \centering
    \includegraphics[width=\linewidth]{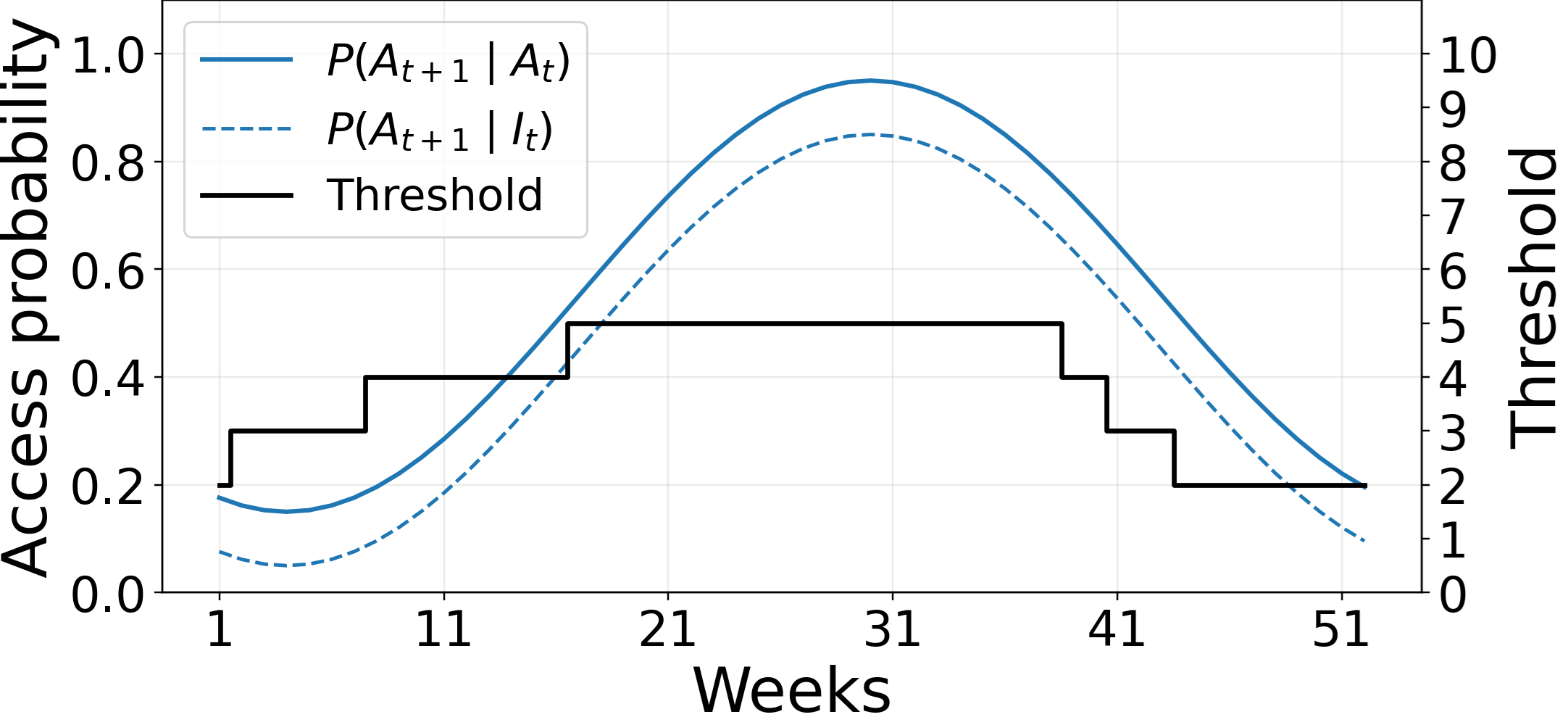}
    \small (e) $\mu=60, \ \sigma=35$
\end{minipage}
\hfill
\begin{minipage}[t]{0.48\linewidth}
    \centering
    \includegraphics[width=\linewidth]{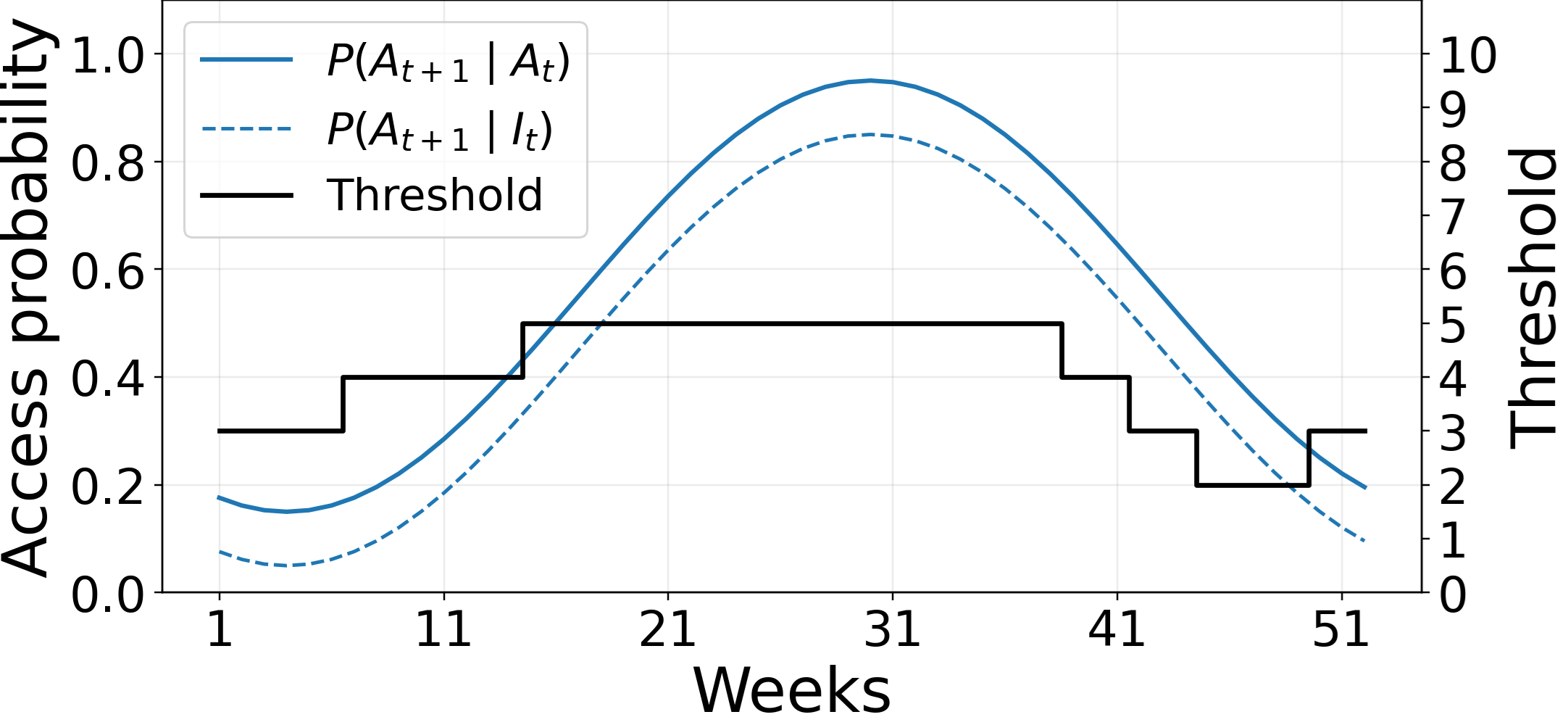}
    \small (f) $\mu=80, \ \sigma=45$
\end{minipage}

\caption{Sensitivity of the optimal maintenance thresholds to degradation characteristics. The left column varies mean time to failure $\mu$, while the right column varies standard deviation of time to failure $\sigma$.}
\label{fig:degradation_sensitivity}
\end{figure}

The standard deviation of time to failure captures uncertainty in degradation timing, which affects the risk of waiting for future maintenance opportunities. As the standard deviation of time to failure increases, the optimal threshold varies more across the seasonal accessibility cycle. Greater degradation uncertainty makes the value of future maintenance opportunities more consequential, causing the policy to respond more strongly to changes in accessibility.

\subsubsection{Economic Characteristics}

Maintenance timing involves a tradeoff between intervening early, which incurs preventive maintenance costs and sacrifices the asset’s remaining useful life, and postponing intervention, which increases the risk of more costly failure. The cost of preventive maintenance therefore plays a key role in determining the optimal threshold. Figure~\ref{fig:economic_sensitivity} shows that lower preventive maintenance costs lead to lower optimal thresholds, while higher preventive maintenance costs shift the thresholds toward more degraded states. As preventive maintenance becomes more costly, its economic advantage over corrective maintenance narrows, so the policy delays intervention and allows the asset to degrade further before maintenance becomes economical. 

\begin{figure}[htbp]
\centering

\begin{minipage}[t]{0.48\linewidth}
    \centering
    \includegraphics[width=\linewidth]{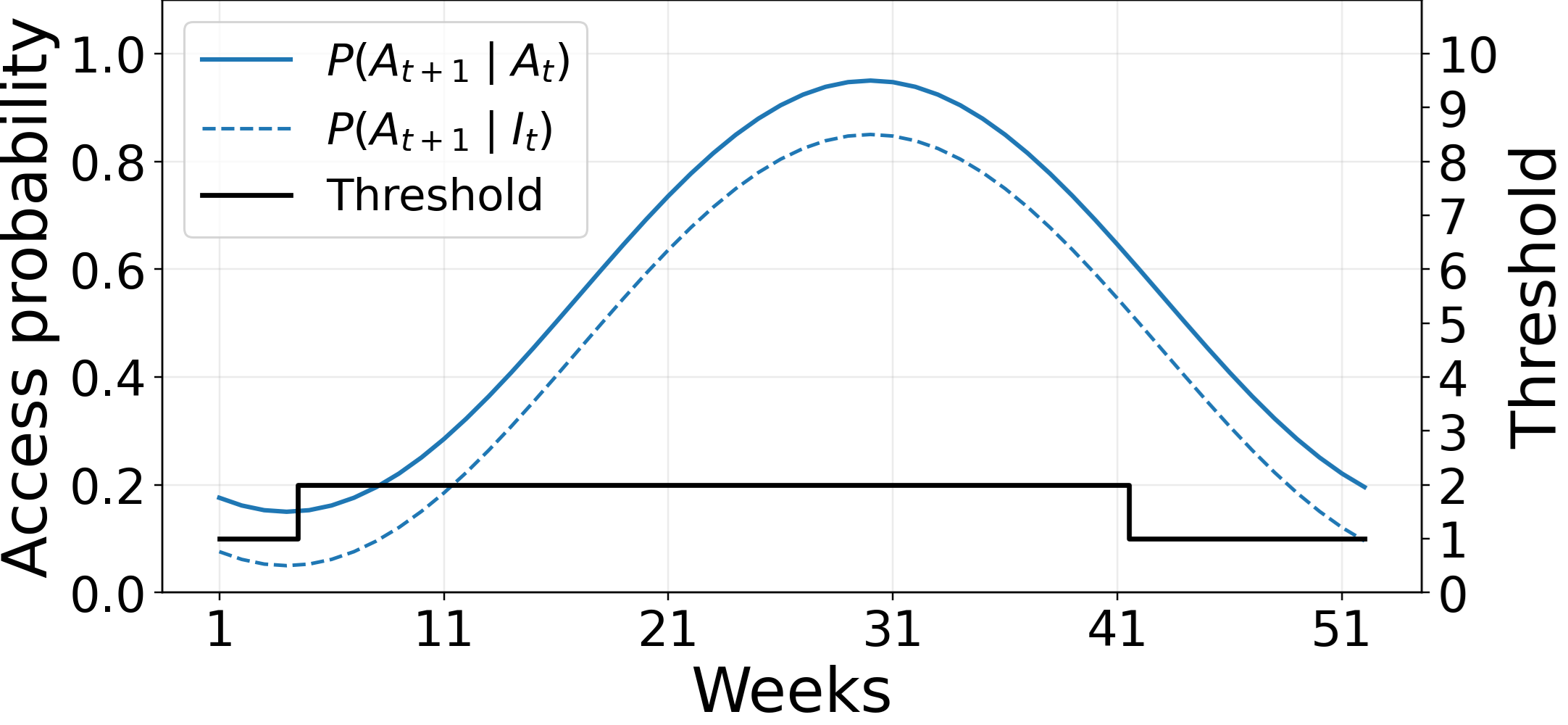}
    \small (a) $c^{PM}=\$25{,}000, \ \gamma_g=1.00$
\end{minipage}
\hfill
\begin{minipage}[t]{0.48\linewidth}
    \centering
    \includegraphics[width=\linewidth]{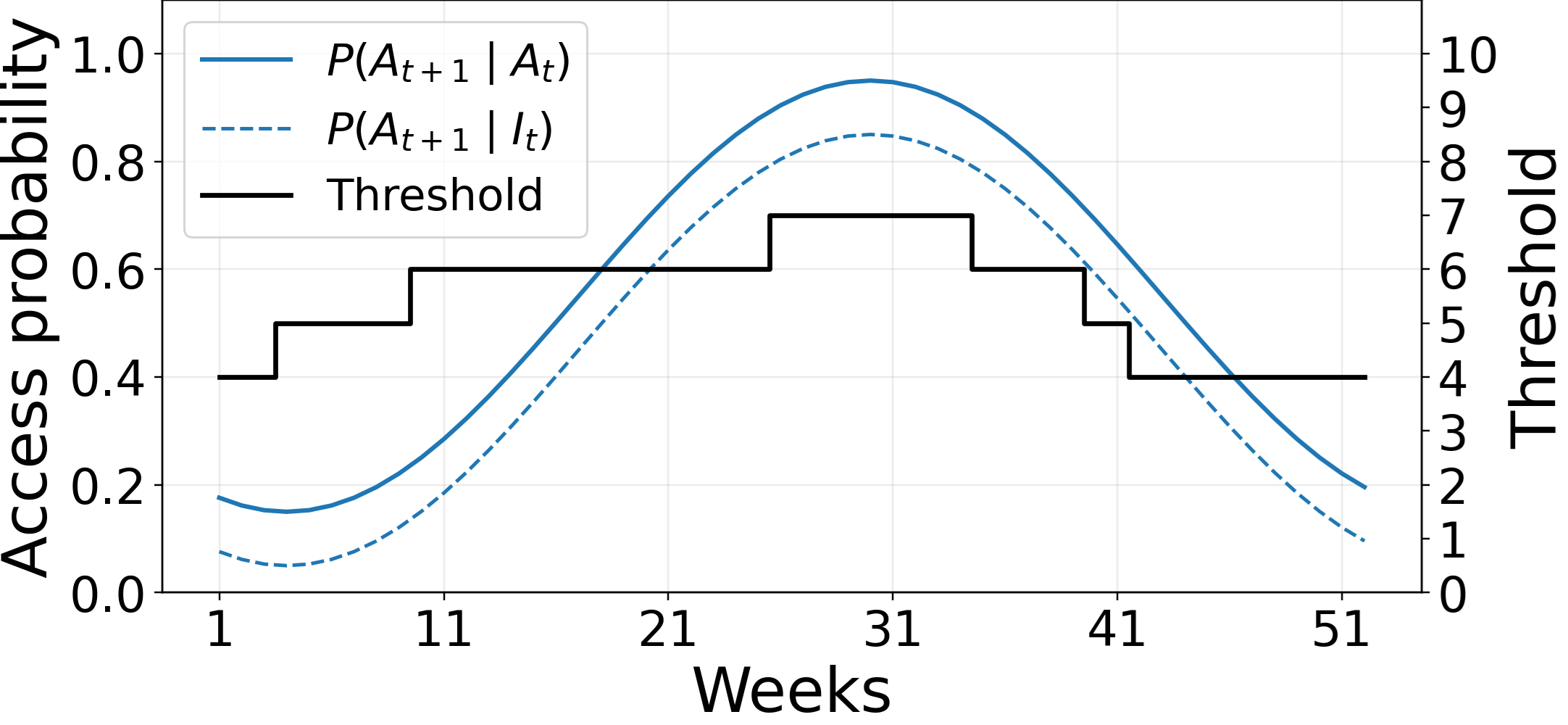}
    \small (b) $c^{PM}=\$150{,}000, \ \gamma_g=0.25$
\end{minipage}

\vspace{2mm}

\begin{minipage}[t]{0.48\linewidth}
    \centering
    \includegraphics[width=\linewidth]{figures/base_case.png}
    \small (c) $c^{PM}=\$150{,}000, \ \gamma_g=1.00$
\end{minipage}
\hfill
\begin{minipage}[t]{0.48\linewidth}
    \centering
    \includegraphics[width=\linewidth]{figures/base_case.png}
    \small (d) $c^{PM}=\$150{,}000, \ \gamma_g=1.00$
\end{minipage}

\vspace{2mm}

\begin{minipage}[t]{0.48\linewidth}
    \centering
    \includegraphics[width=\linewidth]{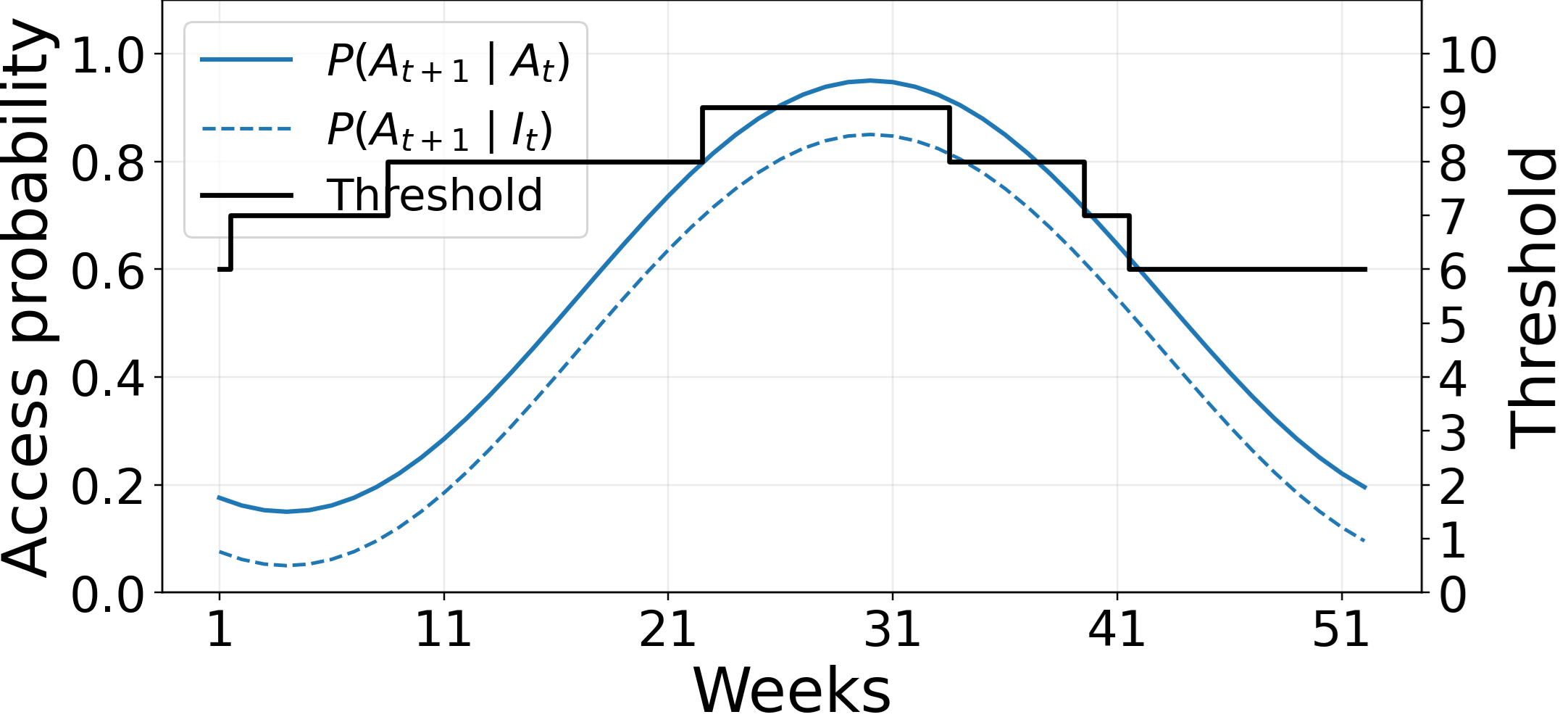}
    \small (e) $c^{PM}=\$749{,}999, \ \gamma_g=1.00$
\end{minipage}
\hfill
\begin{minipage}[t]{0.48\linewidth}
    \centering
    \includegraphics[width=\linewidth]{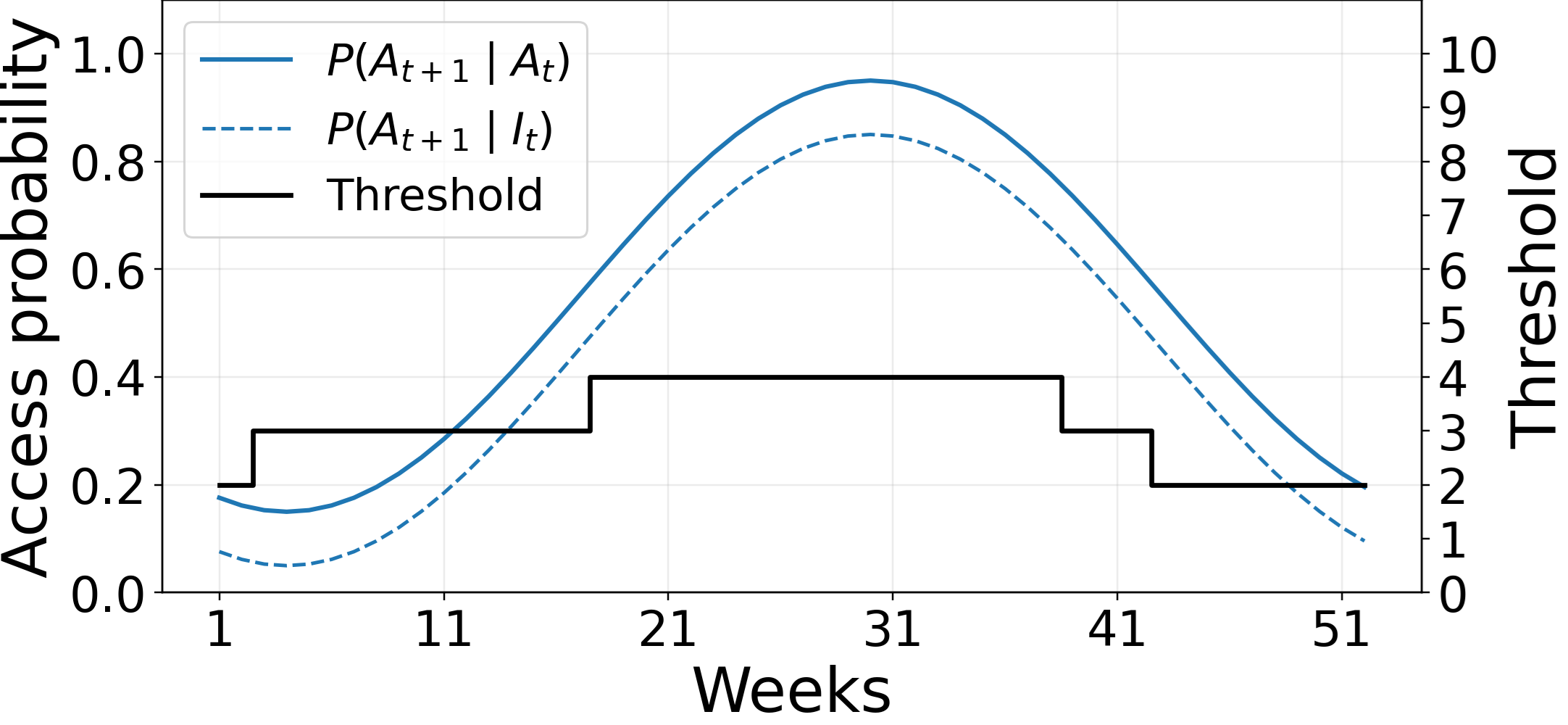}
    \small (f) $c^{PM}=\$150{,}000, \ \gamma_g=4.00$
\end{minipage}

\caption{Sensitivity of the optimal maintenance thresholds to economic characteristics. The left column varies preventive maintenance cost $c^{PM}$, while the right column varies efficiency loss multiplier $\gamma_g$.}
\label{fig:economic_sensitivity}
\end{figure}

Efficiency loss determines the economic consequences of continuing to operate an asset in a degraded condition. As the condition-dependent efficiency-loss multiplier $\gamma_g$ increases, the optimal thresholds decrease. Higher degradation-related operating losses make continued operation in degraded states more costly, increasing the value of earlier preventive maintenance.

\section{Conclusion}
\label{conclusion}

This paper studies condition-based maintenance (CBM) when maintenance opportunities are not continuously available but instead evolve stochastically over time. In such settings, the maintenance decision involves more than determining whether an asset is sufficiently degraded to justify intervention. The decision-maker must also determine whether a currently available maintenance opportunity should be used or whether intervention can be postponed in anticipation of future opportunities. We formulate this problem as a finite-state Markov decision process under a long-run average cost criterion. We analytically characterize the structure of the optimal policy and conduct numerical experiments to quantify the value of condition and accessibility information through comparisons with age-based and constant-threshold policies. We further examine how changes in accessibility, degradation, and cost parameters affect the optimal policy through sensitivity analysis.

Our structural analysis extends the familiar threshold result in CBM to settings with stochastic maintenance opportunities. When maintenance is continuously available, classical CBM models commonly yield a single condition threshold separating continued operation from intervention. In our setting, this structure is preserved, but the threshold becomes accessibility-dependent: for each accessible state, there exists a condition level above which maintenance is optimal. Thus, two assets in the same degradation state may optimally receive different maintenance decisions because the value of waiting depends on the availability of future maintenance opportunities. The result therefore retains the simplicity and interpretability of a threshold policy while incorporating the forward-looking effect of stochastic accessibility. This distinction also sharpens the role of weather relative to earlier weather-aware maintenance models. Prior studies have considered weather through maintenance disruptions, waiting times, seasonal operating conditions, and production losses. Our analysis isolates accessibility as the decision-relevant mechanism and shows structurally how it modifies the standard CBM threshold: current accessibility determines whether maintenance can be performed, while the stochastic evolution of accessibility determines the value of postponing an available intervention. Consequently, the same current level of accessibility need not imply the same maintenance threshold when the outlook for future opportunities differs.

The numerical study illustrates these mechanisms in an offshore wind setting with stochastic and seasonal weather-driven accessibility. The results show that maintenance thresholds vary systematically over the seasonal accessibility cycle and that the same pointwise accessibility level can imply different maintenance decisions depending on the accessibility conditions expected to follow. The benchmark comparisons further quantify the value of the information used by the policy. In the base case, adapting the condition threshold to accessibility reduces long-run average cost by 4.31\% relative to the optimized constant-threshold policy. Compared with the optimized age-based policy, the optimal policy reduces long-run average cost by 33.95\%, while requiring fewer preventive interventions and experiencing substantially fewer corrective interventions. These comparisons distinguish two sources of value: condition information determines whether maintenance is warranted by the physical state of the asset, while accessibility information determines whether a current maintenance opportunity should be used.

The sensitivity analysis further shows that the value and timing of intervention depend jointly on accessibility, degradation, and economic characteristics. Stronger weather persistence increases the differentiation between favorable and unfavorable accessibility periods, while stronger seasonality creates greater variation in maintenance thresholds over the year. Faster degradation leads to earlier intervention, whereas greater uncertainty in time to failure makes the threshold more responsive to changing accessibility conditions. Economic parameters have similarly intuitive effects: lower preventive maintenance costs and higher degradation-related operating losses favor earlier intervention, while costly preventive maintenance and lower operating losses support postponement. Together, the sensitivity and benchmark analysis show that the optimal threshold is not universal: it changes with accessibility, degradation, and economic conditions, and failing to adapt the threshold to these conditions can increase long-run cost. 

Several extensions provide directions for future research. These include multiple interacting assets competing for limited maintenance resources, maintenance durations spanning multiple periods, imperfect maintenance, 
and partially observed asset condition. Incorporating these features would allow the value of accessibility-aware condition-based maintenance to be studied in richer operational environments while building on the structural insights developed in the current study.



\bibliographystyle{pomsref}

 \let\oldbibliography\thebibliography
 \renewcommand{\thebibliography}[1]{%
    \oldbibliography{#1}%
    \baselineskip14pt 
    \setlength{\itemsep}{10pt}
 }
\bibliography{ref1}

@article{ulukus2012optimal,
  title={Optimal replacement policies under environment-driven degradation},
  author={Ulukus, M Yasin and Kharoufeh, Jeffrey P and Maillart, Lisa M},
  journal={Probability in the Engineering and Informational Sciences},
  volume={26},
  number={3},
  pages={405--424},
  year={2012},
  publisher={Cambridge University Press}
}

@article{byon2010optimal,
  title={Optimal maintenance strategies for wind turbine systems under stochastic weather conditions},
  author={Byon, Eunshin and Ntaimo, Lewis and Ding, Yu},
  journal={IEEE Transactions on Reliability},
  volume={59},
  number={2},
  pages={393--404},
  year={2010},
  publisher={IEEE}
}

@article{byon2010season,
  title={Season-dependent condition-based maintenance for a wind turbine using a partially observed Markov decision process},
  author={Byon, Eunshin and Ding, Yu},
  journal={IEEE Transactions on Power Systems},
  volume={25},
  number={4},
  pages={1823--1834},
  year={2010},
  publisher={IEEE}
}

@article{byon2013wind,
  title={Wind turbine operations and maintenance: a tractable approximation of dynamic decision making},
  author={Byon, Eunshin},
  journal={IIE Transactions},
  volume={45},
  number={11},
  pages={1188--1201},
  year={2013},
  publisher={Taylor \& Francis}
}

@article{zheng2020optimal,
  title={Optimal preventive maintenance for wind turbines considering the effects of wind speed},
  author={Zheng, Rui and Zhou, Yifan and Zhang, Yingzhi},
  journal={Wind Energy},
  volume={23},
  number={11},
  pages={1987--2003},
  year={2020},
  publisher={Wiley Online Library}
}

@article{si2018optimal,
  title={An optimal condition-based replacement method for systems with observed degradation signals},
  author={Si, Xiaosheng and Li, Tianmei and Zhang, Qi and Hu, Xiaoxiang},
  journal={IEEE Transactions on Reliability},
  volume={67},
  number={3},
  pages={1281--1293},
  year={2018},
  publisher={IEEE}
}

@article{zhang2013optimal,
  title={Optimal maintenance policy for multi-component systems under Markovian environment changes},
  author={Zhang, Zhuoqi and Wu, Su and Li, Binfeng and Lee, Seungchul},
  journal={Expert Systems with Applications},
  volume={40},
  number={18},
  pages={7391--7399},
  year={2013},
  publisher={Elsevier}
}

@article{barlow1960optimum,
  title={Optimum preventive maintenance policies},
  author={Barlow, Richard and Hunter, Larry},
  journal={Operations Research},
  volume={8},
  number={1},
  pages={90--100},
  year={1960},
  publisher={Informs}
}

@article{mccall1965maintenance,
  title={Maintenance policies for stochastically failing equipment: a survey},
  author={McCall, John J},
  journal={Management Science},
  volume={11},
  number={5},
  pages={493--524},
  year={1965},
  publisher={INFORMS}
}

@article{pierskalla1976survey,
  title={A survey of maintenance models: the control and surveillance of deteriorating systems},
  author={Pierskalla, William P and Voelker, John A},
  journal={Naval Research Logistics Quarterly},
  volume={23},
  number={3},
  pages={353--388},
  year={1976},
  publisher={Wiley Online Library}
}

@article{sherif1981optimal,
  title={Optimal maintenance models for systems subject to failure--a review},
  author={Sherif, YS and Smith, ML},
  journal={Naval Research Logistics Quarterly},
  volume={28},
  number={1},
  pages={47--74},
  year={1981},
  publisher={Wiley Online Library}
}

@article{deJongeScarf2020review,
  title={A review on maintenance optimization},
  author={De Jonge, Bram and Scarf, Philip A},
  journal={European Journal of Operational Research},
  volume={285},
  number={3},
  pages={805--824},
  year={2020},
  publisher={Elsevier}
}

@book{puterman2014markov,
  title={Markov Decision Processes: Discrete Stochastic Dynamic Programming},
  author={Puterman, Martin L},
  year={2014},
  publisher={John Wiley \& Sons}
}

@article{jardine2006review,
  title={A review on machinery diagnostics and prognostics implementing condition-based maintenance},
  author={Jardine, Andrew KS and Lin, Daming and Banjevic, Dragan},
  journal={Mechanical Systems and Signal Processing},
  volume={20},
  number={7},
  pages={1483--1510},
  year={2006},
  publisher={Elsevier}
}

@article{alaswad2017review,
  title={A review on condition-based maintenance optimization models for stochastically deteriorating system},
  author={Alaswad, Suzan and Xiang, Yisha},
  journal={Reliability Engineering \& System Safety},
  volume={157},
  pages={54--63},
  year={2017},
  publisher={Elsevier}
}

@article{elwany2011structured,
  title={Structured replacement policies for components with complex degradation processes and dedicated sensors},
  author={Elwany, Alaa H and Gebraeel, Nagi Z and Maillart, Lisa M},
  journal={Operations Research},
  volume={59},
  number={3},
  pages={684--695},
  year={2011},
  publisher={INFORMS}
}

@article{zhu2019dynamic,
  title={A dynamic programming-based maintenance model of offshore wind turbine considering logistic delay and weather condition},
  author={Zhu, Wenjin and Castanier, Bruno and Bettayeb, Belgacem},
  journal={Reliability Engineering \& System Safety},
  volume={190},
  pages={106512},
  year={2019},
  publisher={Elsevier}
}

@article{shafiee2019maintenance,
  title={Maintenance optimization and inspection planning of wind energy assets: Models, methods and strategies},
  author={Shafiee, Mahmood and S{\o}rensen, John Dalsgaard},
  journal={Reliability Engineering \& System Safety},
  volume={192},
  pages={105993},
  year={2019},
  publisher={Elsevier}
}

@article{ren2021offshore,
  title={Offshore wind turbine operations and maintenance: A state-of-the-art review},
  author={Ren, Zhengru and Verma, Amrit Shankar and Li, Ye and Teuwen, Julie JE and Jiang, Zhiyu},
  journal={Renewable and Sustainable Energy Reviews},
  volume={144},
  pages={110886},
  year={2021},
  publisher={Elsevier}
}

@article{zhang2019opportunistic,
  title={Opportunistic maintenance strategy for wind turbines considering weather conditions and spare parts inventory management},
  author={Zhang, Chen and Gao, Wei and Yang, Tao and Guo, Sheng},
  journal={Renewable Energy},
  volume={133},
  pages={703--711},
  year={2019},
  publisher={Elsevier}
}

@article{de2019discretizing,
  title={Discretizing continuous-time continuous-state deterioration processes, with an application to condition-based maintenance optimization},
  author={de Jonge, Bram},
  journal={Reliability Engineering \& System Safety},
  volume={188},
  pages={1--5},
  year={2019},
  publisher={Elsevier}
}

@article{sun2023robust,
  title={Robust condition-based production and maintenance planning for degradation management},
  author={Sun, Qiuzhuang and Chen, Piao and Wang, Xin and Ye, Zhi-Sheng},
  journal={Production and Operations Management},
  volume={32},
  number={12},
  pages={3951--3967},
  year={2023},
  publisher={SAGE Publications Sage CA: Los Angeles, CA}
}

@article{liu2017condition,
  title={A condition-based maintenance policy for degrading systems with age-and state-dependent operating cost},
  author={Liu, Bin and Wu, Shaomin and Xie, Min and Kuo, Way},
  journal={European Journal of Operational Research},
  volume={263},
  number={3},
  pages={879--887},
  year={2017},
  publisher={Elsevier}
}

@article{staffell2014does,
  title={How does wind farm performance decline with age?},
  author={Staffell, Iain and Green, Richard},
  journal={Renewable Energy},
  volume={66},
  pages={775--786},
  year={2014},
  publisher={Elsevier}
}

@article{dao2021integrated,
  title={Integrated condition-based maintenance modelling and optimisation for offshore wind turbines},
  author={Dao, Cuong D and Kazemtabrizi, Behzad and Crabtree, Christopher J and Tavner, Peter J},
  journal={Wind Energy},
  volume={24},
  number={11},
  pages={1180--1198},
  year={2021},
  publisher={Wiley Online Library}
}

@article{raza2019optimal,
  title={Optimal preventive maintenance of wind turbine components with imperfect continuous condition monitoring},
  author={Raza, Ahmed and Ulansky, Vladimir},
  journal={Energies},
  volume={12},
  number={19},
  pages={3801},
  year={2019},
  publisher={MDPI}
}

@article{panagiotidou2010statistical,
  title={Statistical process control and condition-based maintenance: A meaningful relationship through data sharing},
  author={Panagiotidou, Sofia and Tagaras, George},
  journal={Production and Operations Management},
  volume={19},
  number={2},
  pages={156--171},
  year={2010},
  publisher={Wiley Online Library}
}

@article{hagen2013multivariate,
  title={A multivariate Markov weather model for O\&M simulation of offshore wind parks},
  author={Hagen, Brede and Simonsen, Ingve and Hofmann, Matthias and Muskulus, Michael},
  journal={Energy Procedia},
  volume={35},
  pages={137--147},
  year={2013},
  publisher={Elsevier}
}

@article{martini2017accessibility,
  title={Accessibility assessment for operation and maintenance of offshore wind farms in the North Sea},
  author={Martini, Michele and Guanche, Ra{\'u}l and Losada, I{\~n}igo J and Vidal, C{\'e}sar},
  journal={Wind Energy},
  volume={20},
  number={4},
  pages={637--656},
  year={2017},
  publisher={Wiley Online Library}
}

@article{drent2024condition,
  title={Condition-based production for stochastically deteriorating systems: Optimal policies and learning},
  author={Drent, Collin and Drent, Melvin and Arts, Joachim},
  journal={Manufacturing \& Service Operations Management},
  volume={26},
  number={3},
  pages={1137--1156},
  year={2024},
  publisher={INFORMS}
}

@article{sloan2000combined,
  title={Combined production and maintenance scheduling for a multiple-product, single-machine production system},
  author={Sloan, Thomas W and Shanthikumar, J George},
  journal={Production and Operations Management},
  volume={9},
  number={4},
  pages={379--399},
  year={2000},
  publisher={Wiley Online Library}
}

@article{wang2025learning,
  title={Learning to balance the performance and deterioration of aging systems through derating},
  author={Wang, Jue},
  journal={Production and Operations Management},
  volume={34},
  number={7},
  pages={1743--1758},
  year={2025},
  publisher={SAGE Publications Sage CA: Los Angeles, CA}
}

@article{shahri2026data,
  title={A data-driven robust approach to a problem of optimal replacement in maintenance},
  author={Shahri Majarshin, Sina and Marandi, Ahmadreza and Fecarotti, Claudia and van Houtum, Geert-Jan},
  journal={Annals of Operations Research},
  pages={1--32},
  year={2026},
  publisher={Springer}
}

@article{liang2023reliability,
  title={A reliability model for systems subject to mutually dependent degradation processes and random shocks under dynamic environments},
  author={Liang, Qingzhu and Yang, Yinghao and Peng, Changhong},
  journal={Reliability Engineering \& System Safety},
  volume={234},
  pages={109165},
  year={2023},
  publisher={Elsevier}
}

@article{li2022failure,
  title={Failure rate assessment for onshore and floating offshore wind turbines},
  author={Li, He and Peng, Weiwen and Huang, Cheng-Geng and Guedes Soares, C},
  journal={Journal of Marine Science and Engineering},
  volume={10},
  number={12},
  pages={1965},
  year={2022},
  publisher={MDPI}
}

@article{dao2019wind,
  title={Wind turbine reliability data review and impacts on levelised cost of energy},
  author={Dao, Cuong and Kazemtabrizi, Behzad and Crabtree, Christopher},
  journal={Wind Energy},
  volume={22},
  number={12},
  pages={1848--1871},
  year={2019},
  publisher={Wiley Online Library}
}

@article{kiadaliry2026weibull,
  title={Weibull-Neural Network Framework for Wind Turbine Lifetime Monitoring and Disturbance Identification},
  author={Kiadaliry, Fatemeh and Raissi, Sadigh and Komijan, Alireza Rashidi},
  journal={Wind Energy},
  volume={29},
  number={4},
  pages={e70103},
  year={2026},
  publisher={Wiley Online Library}
}

@article{kolesar1966minimum,
  title={Minimum cost replacement under Markovian deterioration},
  author={Kolesar, Peter},
  journal={Management Science},
  volume={12},
  number={9},
  pages={694--706},
  year={1966},
  publisher={INFORMS}
}

@article{mckone2002guidelines,
  title={Guidelines for implementing predictive maintenance},
  author={McKone, Kathleen E and Weiss, Elliott N},
  journal={Production and Operations Management},
  volume={11},
  number={2},
  pages={109--124},
  year={2002},
  publisher={Wiley Online Library}
}

@article{batun2012reassessing,
  title={Reassessing tradeoffs inherent to simultaneous maintenance and production planning},
  author={Batun, Sakine and Maillart, Lisa M},
  journal={Production and Operations Management},
  volume={21},
  number={2},
  pages={396--403},
  year={2012},
  publisher={SAGE Publications Sage CA: Los Angeles, CA}
}




\end{document}